\documentclass[journal]{IEEEtran}

\usepackage{hyperref}
\usepackage{color}
\usepackage{amsmath}
\usepackage{amssymb}
\usepackage{graphicx}
\usepackage{amsthm}
\usepackage{cite}
\usepackage[ruled,linesnumbered]{algorithm2e}
\usepackage{multirow}
\usepackage{diagbox}
\usepackage{subfigure}
\usepackage{cleveref}

\newtheorem{proposition}{Proposition}
\newtheorem{remark}{Remark}
\newtheorem{lemma}{Lemma}

\begin{document}
%
\title{Joint Channel Estimation and Data Recovery for Millimeter Massive MIMO: Using Pilot to Capture Principal Components}
%
%
%

\author{Shusen~Cai,
        Li~Chen~\IEEEmembership{Senior~Member~IEEE},
        Yunfei~Chen~\IEEEmembership{Senior~Member~IEEE} 

        Huarui~Yin~\IEEEmembership{Member~IEEE}, 
        and Weidong Wang
\thanks{Shusen Cai, Li Chen, Huarui Yin and Weidong Wang are with the Department of Electronic Engineering and Information Science, University of Science and Technology of China, Hefei 230027, Anhui, China (e-mail: cssemail@mail.ustc.edu.cn; chenli87@ustc.edu.cn; yhr@ustc.edu.cn; wdwang@ustc.edu.cn).} 
\thanks{Yunfei Chen is with the Department of Engineering, University of Durham, Durham DH1 3LE, U.K. (e-mail: yunfei.chen@durham.ac.uk).}}

\maketitle

\begin{abstract}
Channel state information (CSI) is important to reap the full benefits of millimeter wave (mmWave) massive multiple-input multiple-output (MIMO) systems. The traditional channel estimation methods using pilot frames (PF) lead to excessive overhead. To reduce the demand for PF, data frames (DF) can be adopted for joint channel estimation and data recovery. However, the computational complexity of the DF-based methods is prohibitively high. To reduce the computational complexity, we propose a joint channel estimation and data recovery (JCD) method assisted by a small number of PF for mmWave massive MIMO systems. The proposed method has two stages. In Stage 1, differing from the traditional PF-based methods, the proposed PF-assisted method is utilized to capture the angle of arrival (AoA) of principal components (PC) of channels. In Stage 2, JCD is designed for parallel implementation based on the multi-user decoupling strategy. The theoretical analysis demonstrates that the PF-assisted JCD method can achieve equivalent performance to the Bayesian-optimal DF-based method, while greatly reducing the computational complexity. Simulation results are also presented to validate the analytical results.
\end{abstract}


%
\IEEEpeerreviewmaketitle

\section{Introduction}\label{section_1}
\IEEEPARstart{M}{illimeter} wave (mmWave) massive multiple-input multiple-output (MIMO) has been considered as a key enabler for future cellular wireless communications \cite{HuangJie2019}. It can greatly increase information transmission rates, enhance spectral efficiencies and provide spatial diversity using the joint capabilities of its ultra-wide bandwidth in the mmWave frequency bands (30-300GHz) \cite{Rappaport2017,YangXi2019} and high beamforming gains achieved by massive antenna arrays \cite{Albreem2019,Busari2018}.

The attainment of the most benefits above requires the channel state information (CSI). The acquisition of CSI is traditionally accomplished through estimation from pilot frames (PF) (i.e. PF-based methods). The classical PF-based methods include the least squares (LS) method \cite{Ozdemir2007} and the minimum mean square error (MMSE) method \cite{YinHaifan2013}. These methods employ linear matrix operations, and are applicable for Rayleigh fading MIMO channels. For mmWave massive MIMO systems, the oversight of sparsity of channels incurs excessive PF overhead in these methods.

By exploiting the sparsity of mmWave massive MIMO channels, compressive sensing is adopted to reduce PF overhead \cite{TaoJun2019,Duong2023,Flinth2016,LianLixiang2018,Hawej2019,LiuAn2016,WeiChao2017,MoJianhua2018,Rajoriya2023}. In \cite{TaoJun2019,Duong2023}, orthogonal matching pursuit (OMP) based algorithms were proposed for sparse channel estimation, and could simply track the parameters of multi-path components to reconstruct the channel matrix. The estimation of the channel matrix was modeled as a least squares minimization with \(l_{0}\) norm regularization in \cite{Flinth2016,LianLixiang2018,Hawej2019,LiuAn2016}, and could be solved by using relevant optimization algorithms. The approximate message passing (AMP) algorithm was applied to obtain the posterior mean estimator of the channel matrix \cite{WeiChao2017,MoJianhua2018,Rajoriya2023}. 

Besides PF, data frames (DF) can be also used for acquiring CSI, thereby reducing the demand for PF \cite{MaShuo2018,Mezghani2018,Ghavami2018,ZhangJianwen2018,Xiong2019,Wen2016,ZhangRuoyu2022,Du2021,DuJianhe2023}. In \cite{MaShuo2018,Mezghani2018}, the expectation maximization (EM) algorithm was utilized to design a DF-based channel estimator for obtaining the modulus values of channel parameters. By utilizing expectation propagation (EP) algorithm, the DF-based method proposed in \cite{Ghavami2018} further improved estimation accuracy based on the estimation results from the eigenvalue decomposition method. In \cite{ZhangJianwen2018,Xiong2019,Wen2016}, the DF-based methods adopting bilinear generalized approximate message passing (BiGAMP) algorithm exclusively utilized basic operations (i.e. additions and multiplications) for iterative message updates, and could converge to the Bayesian optimal performance of joint channel estimation and data recovery.

The joint recovery of both channel parameters and data symbols in the existing DF-based methods \cite{MaShuo2018,Mezghani2018,Ghavami2018,ZhangJianwen2018,Xiong2019,Wen2016} leads to high computational complexity, especially for mmWave massive MIMO systems. To simplify the joint recovery process, DF-based methods utilizing tensor decomposition were proposed in \cite{ZhangRuoyu2022,Du2021,DuJianhe2023}. The tensor decomposition algorithm treated the observed signals as the product of the unknown symbol tensor and unknown channel matrix tensor, and estimated tensors using the alternating least squares (ALS) method. The mmWave massive MIMO channel matrix could be simply reconstructed by the parameters of multi-path components extracted from the estimated channel matrix tensor.

However, tensor decomposition algorithm can't resolve the high computational complexity issue caused by large-scale antenna arrays, because the computational complexity of the alternating tensor estimation grows quadratically with the number of antennas. To reduce the computational complexity of estimating the high-dimensional channel matrix, \cite{ChenXiao2021,WangAnding2019,ShanYaru2021} represented the channel matrix in the angular (and delay) domain, and utilized a small number of multi-path components to simplify the PF-based estimation of high-dimensional channel matrix. This has inspired us to use the principal components (PC) in mmWave massive MIMO channels to simplify the process of joint channel estimation and data recovery.

Motivated by the above observations, this work proposes a joint channel estimation and data recovery (JCD) method assisted by a limited number of PF for a low computational complexity design. The proposed PF-assisted JCD method involves two stages: coarse channel estimation stage and joint fine channel estimation and data recovery stage. In Stage 1, differing from the traditional PF-based methods, the PF-assisted coarse channel estimation is utilized to preliminarily acquire the prior knowledge of massive MIMO channels in the angular domain. Based on the acquired prior knowledge, the search range of the angle of arrival (AoA) is confined to focus on the PC of channels. Then, all users are decoupled into user groups without mutual interference (i.e. no multi-path component overlap) by the multi-user decoupling strategy. In Stage 2, DF are utilized to realize the JCD of all user groups in parallel using the decoupled EM-BiGAMP algorithm. The theoretical analysis shows that the proposed PF-assisted JCD method can reduce computational complexity without performance loss, which is also confirmed by simulation. The main contributions of this paper are summarized as follows:

\begin{itemize}
\item \textbf{A low computational complexity PF-assisted JCD method:} The low computational complexity method is based on the PC acquisition and block-wise partitioning of the high-dimensional channel matrix, and obtained by a two-stage estimation. In Stage 1, the proposed PF-assisted method aims to capture the PC of channels within an extremely narrow AoA search range, which significantly reduces the dimensionality of the observed signals. To further accelerate
the solution of the multi-user estimation problem, the designed JCD algorithm (i.e decoupled EM-BiGAMP algorithm) partitions the channel matrix, and then parallelly implements the solution to the estimation problems for submatrices in Stage 2.

\item \textbf{Performance and complexity analysis of PF-assisted JCD method:} Mean squared error (MSE) is employed as a performance metric for studying the proposed PF-assisted JCD method. According to both theoretical analysis and simulations, the sparsity-exploiting operations of the proposed PF-assisted JCD method are verified not to cause performance loss in the joint channel estimation and data recovery, and the computational complexity and processing time of the PF-assisted JCD method are respectively reduced by approximately \(\alpha_{\rm r}\) times and approximately \(\alpha_{\rm r}\cdot N\) times, compared to the original DF-based method. \(N\) is the number of users and \(\alpha_{\rm r}\) is the proportion of signal dimension reduction.

\end{itemize}

The rest of the paper is organized as follows. In Section \ref{section_2}, we will describe the mmWave massive MIMO system in an uplink communication scenario, and raise the issue of the excessive computational complexity in existing channel estimation methods. In Sections \ref{section_3}, we will propose the PF-assisted JCD method for a low computational complexity design. In Sections \ref{section_4}, the MSE performance analysis of the proposed PF-assisted JCD method will be provided. Simulation results will be shown in Section \ref{section_5}, followed by the conclusions in Section \ref{section_6}.

\textit{Notations:} \(a\) is a scalar, \(\boldsymbol{\rm a}\) is a vector, \(\boldsymbol{\rm A}\) is a matrix. \(\boldsymbol{\rm I}_{a}\) is an \(a \times a\) identity matrix. \({\mathbb{C}}\) represents a complex number field. \({\mathcal{CN}}(\boldsymbol{\rm x};\boldsymbol{\rm a},\boldsymbol{\rm A})\) means that random vector \(\boldsymbol{\rm x}\) obeys a Gaussian distribution with mean \(\boldsymbol{\rm a}\) and covariance matrix \(\boldsymbol{\rm A}\), and \({\delta(\cdot) } \) is a dirac delta function. \({\rm E}\{\cdot\}\) and \({\rm Var}\{\cdot\}\) are the calculations of mathematical expectation and variance respectively. \((\cdot)^T\) is defined as the transpose of a matrix or a vector, \((\cdot)^H\) is the conjugate transpose of a matrix or a vector, and \({\rm det}(\cdot)\) is the determinant operation of a matrix. \(\left\| \cdot \right\|\) is the Euclidean norm of a matrix or a vector.

\section{System Model and Problem Formulation}\label{section_2}

\subsection{System model}\label{section_2_1}
Consider an uplink communication scenario, where \(N\) single-antenna users are simultaneously served by a base station (BS) equipped with a uniform linear array (ULA) of \(M\) antennas, as illustrated in Fig. \ref{system_model}. The BS receives the uplink signals transmitted by users, and the radio frequency (RF) chains and analog combiner \(\boldsymbol{\rm F}\) are used to process the uplink signals as
\begin{equation}\label{original_signal}
{\boldsymbol{\mathrm{y}}_k} = \boldsymbol{\rm F}(\boldsymbol{\mathrm{G}}{\boldsymbol{\mathrm{x}}_k} + {\boldsymbol{\mathrm{w}}_k}) = \boldsymbol{\mathrm{H}}{\boldsymbol{\mathrm{x}}_k} + {\boldsymbol{\mathrm{n}}_k},\  \forall k \in \{ 1,2,...,K\} 
\end{equation}
where \(\boldsymbol{\mathrm{G}} \in {{\mathbb{C}}^{M \times N}}\) is a mmWave channel matrix, \(\boldsymbol{\mathrm{x}}_{k} \in {{\mathbb{C}}^{N \times 1}}\) is a multi-user symbol vector at the \(k\)th symbol duration, \(\boldsymbol{\mathrm{w}}_{k} \in {{\mathbb{C}}^{M \times 1}}\) is the additive white Gaussian noise (AWGN) vector with zero mean and covariance matrix \(\sigma_{\rm w}^2 \boldsymbol{\rm I}_{M}\) at the \(k\)th symbol duration, and \(K\) is the total number of symbols. After being processed by the analog combiner \(\boldsymbol{\rm F}\), the equivalent channel matrix and equivalent noise vector can be represented as \( \boldsymbol{\rm H} = \boldsymbol{\rm F} \boldsymbol{\rm G} \) and \({\boldsymbol{\mathrm{n}}_k} = {\boldsymbol{\rm F}}{\boldsymbol{\mathrm{w}}_k}\) respectively. 

Due to the highly directional, quasi-optical nature of mmWave channels, the channel matrix \(\boldsymbol{\mathrm{G}} = [\boldsymbol{\mathrm{g}}_{1}, \boldsymbol{\mathrm{g}}_{2},\ ...\ ,\boldsymbol{\mathrm{g}}_{N}]\) can be modeled as
\begin{equation}
{\boldsymbol{\mathrm{g}}_n} = \sum\limits_{l = 1}^{{L_n}} {{\beta _{n,l}}} {\boldsymbol{\mathrm{a}}_M}({\theta _{n,l}}),\quad\forall n \in \{ 1,2,...,N\}
\end{equation}
where 
\begin{equation}
\begin{aligned}
{{\boldsymbol{\mathrm{a}}}_M}(\theta ) = \left[ {1,{e^{ - j\cdot \pi \sin \theta }},{e^{ - j \cdot 2\pi \sin \theta }},{\rm{ }}...{\rm{ }},{e^{ - j\cdot (M - 1) \pi \sin \theta  }}} \right]^{T},
\end{aligned}
\end{equation}
\({\boldsymbol{\mathrm{a}}_M}(\theta)\) is the steering vector of a ULA with half-wavelength antenna spacing. \(\beta _{n,l}\) is the complex gain of the \(l\)th path of the \(n\)th user, and \(\theta _{n,l}\) is the \(l\)th path's AoA of the \(n\)th user. \(L_{n}\) is the total number of multi-path components of the \(n\)th user. 

Transforming the channel matrix \(\boldsymbol{\rm G}\) into the angular domain representation, one has
\begin{equation}
\begin{aligned}
& {\boldsymbol{\mathrm{G}}} = {\boldsymbol{\mathrm{U}}}{\tilde{\boldsymbol{\mathrm{G}}}},
\end{aligned}
\end{equation}
where
\begin{equation}
{\boldsymbol{\mathrm{U}}} = \frac{1}{{\sqrt {{M}} }}\left[ {{\boldsymbol{\mathrm{a}}_{{M}}}(0),{\boldsymbol{\mathrm{a}}_{{M}}}(\frac{1}{{{M}}}),\ ...\ ,{\boldsymbol{\mathrm{a}}_{{M}}}(\frac{{{M} - 1}}{{{M}}})} \right]
\end{equation}
is a normalized discrete Fourier transform (DFT) matrix of \(M\)-point, and \({\tilde{\boldsymbol{\mathrm{G}}}}\) is the channel matrix represented in the angular domain. To receive the multi-path components with unknown AoA, the DFT analog combiner is adopted to give 
\begin{equation}
\boldsymbol{\mathrm{F}} = {\boldsymbol{\mathrm{U}}}^{-1} = {\boldsymbol{\mathrm{U}}}^{H}.    
\end{equation}
where the equivalent channel matrix \(\boldsymbol{\rm H}\) equals to \(\tilde{\boldsymbol{\rm G}}\). Because of the small value of \(L_{n}\), the matrix \(\boldsymbol{\rm H}\) exhibits sparse characteristics, manifested by a large amount of signal energy concentrated on a small number of multi-path components. 

\begin{figure}
\centering
\includegraphics[width=\linewidth]{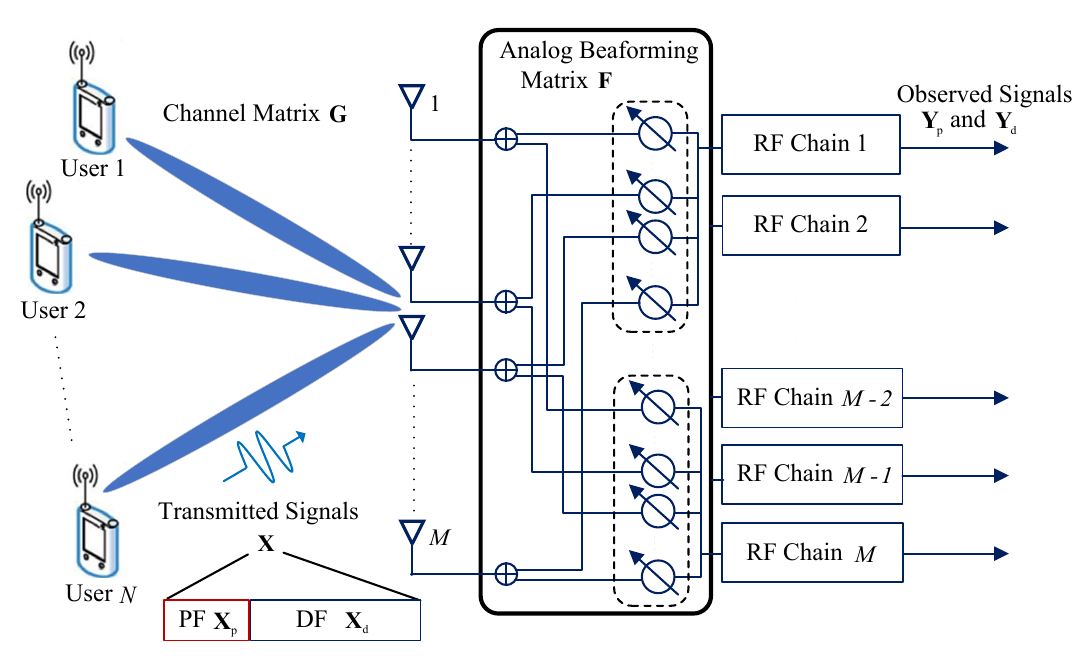}
\caption{System Model: a mmWave massive MIMO uplink communcation system with a BS and multiple users.}\label{system_model} 
\end{figure}

\subsection{Problem formulation}\label{section_2_2}
For MIMO communication systems, the acquisition of the channel matrix \(\boldsymbol{\rm H}\) utilizing the observed signal in \eqref{original_signal} needs to be addressed. The frames of the observed signal \({\boldsymbol{\mathrm{X}}} = [\boldsymbol{\mathrm{x}}_{1}, \boldsymbol{\mathrm{x}}_{2},\ ...\ ,\boldsymbol{\mathrm{x}}_{K}] \in {{\mathbb{C}}^{N \times K}}\) are comprised of PF \({{\boldsymbol{\mathrm{X}}}_{\mathrm{p}}}\in {{\mathbb{C}}^{N \times K_{\mathrm{p}}}}\) and DF \({{\boldsymbol{\mathrm{X}}}_{\mathrm{d}}}\in {{\mathbb{C}}^{N \times K_{\mathrm{d}}}}\) (i.e. \({{\boldsymbol{\mathrm{X}}}} = [{{\boldsymbol{\mathrm{X}}}_{\mathrm{p}}}, {{\boldsymbol{\mathrm{X}}}_{\mathrm{d}}}]\), and \(K = K_{\mathrm{p}} + K_{\mathrm{d}}\)). Based on the types of time frames used, the channel estimation methods are mainly divided into the following two categories:

\subsubsection{\bf{PF-based method}}
The traditional channel estimation methods only utilize PF to acquire channel state information, and the most classical methods are the LS method and the MMSE method. In order to exploit the sparsity of the channel matrix \(\boldsymbol{\rm H}\) to reduce system complexity and enhance estimation performance, many PF-based methods have been already proposed, such as OMP methods \cite{TaoJun2019,Duong2023}, least absolute shrinkage and selection operator (LASSO) methods \cite{Flinth2016,LianLixiang2018,Hawej2019,LiuAn2016}, AMP methods \cite{WeiChao2017,MoJianhua2018,Rajoriya2023}. 

The main disadvantage of these PF-based methods is that their performance depends on the length of PF. Because of the high sensitivity of mmWave channels to changes in the propagation environment, the coherence time of the mmWave channel is relatively short, hence requiring frequent CSI updates \cite{Du2021}. The PF-based channel estimation methods require a significant number of PF sequences to complete a single measurement and thus, is not suitable for frequent CSI updates.

\subsubsection{\bf{DF-based method}}
In order to address the issue of excessive PF overhead in the traditional PF-based methods, the DF-based methods can use both PF and DF for joint channel estimation and data recovery, which is given by
\begin{equation}\label{bilinear_problem}
\left[ {
{{\boldsymbol{\rm Y}_{\rm{p}}}},{{\boldsymbol{\rm Y}_{\rm{d}}}}
} \right] = \boldsymbol{\rm H}\left[ {
{{\boldsymbol{\rm X}_{\rm{p}}}},{{\boldsymbol{\rm X}_{\rm{d}}}}
} \right] + \left[ {
{{\boldsymbol{\rm N}_{\rm{p}}}},{{\boldsymbol{\rm N}_{\rm{d}}}}
} \right],
\end{equation}
where the probability density functions (PDF) of the channel matrix \(\boldsymbol{\rm H}\) and data matrix \(\boldsymbol{\rm X}\) are known (Note that \(\boldsymbol{\rm X}_{\rm p}\) can be regarded as random variables with known mean and zero variance). 

The joint estimation of channel matrix \(\boldsymbol{\rm H}\) and symbol matrix \(\boldsymbol{\rm X}_{\rm d}\) can be modeled as a bilinear inference problem \cite{Parker2016}. The BiGAMP algorithm proposed in \cite{Parker2014} is an excellent method for solving bilinear problems, which simplifies the loopy message passing (MP) algorithm to achieve implementation through simple computations (i.e. addition and multiplication operations), and has been applied in Rayleigh and multi-path fading channels \cite{Xiong2019,Wen2016}. The factor graph of the BiGAMP algorithm can be represented in Fig. \ref{fig:Factor_Graph}. Through iterative updates of messages \(\Delta _{m \to nk}^{\boldsymbol{\rm X}}(t,{x_{n,k}})\), \(\Delta _{k \to mn}^{\boldsymbol{\rm H}}(t,{h_{m,n}})\), \(\Delta _{m \leftarrow nk}^{\boldsymbol{\rm X}}(t,{x_{n,k}})\) and \(\Delta _{k \leftarrow mn}^{\boldsymbol{\rm H}}(t,{h_{m,n}})\), the estimated log-posterior PDF of \(\boldsymbol{\rm X}\) and \(\boldsymbol{\rm H}\), \(\Delta _{nk}^{\boldsymbol{\rm X}}(t,{x_{n,k}})\) and \(\Delta _{mn}^{\boldsymbol{\rm H}}(t,{h_{m,n}})\), can converge to the true log-posterior PDF \({P_{{x_{n,k}}|\boldsymbol{\rm Y}}}({x_{n,k}}|\boldsymbol{\rm Y})\) and \({P_{{h_{m,n}}|\boldsymbol{\rm Y}}}({h_{m,n}}|\boldsymbol{\rm Y})\). Given the estimated posterior PDF, the Bayesian optimal estimator of \(x_{n,k}\) and \(h_{m,n}\) can be obtained by 
\begin{equation}\label{DF_Ori_7}
\hat x_{n,k}^{(t)} = \int_{{x_{n,k}}} {{x_{n,k}} \cdot {P_{{x_{n,k}}|\boldsymbol{\rm Y}}^{(t)}}({x_{n,k}}|\boldsymbol{\rm Y})} {\rm d}{x_{n,k}},
\end{equation}
\begin{equation}\label{DF_Ori_8}
\hat h_{m,n}^{(t)} = \int_{{h_{m,n}}} {{h_{m,n}} \cdot {P_{{h_{m,n}}|\boldsymbol{\rm Y}}^{(t)}}({h_{m,n}}|\boldsymbol{\rm Y})} {\rm d}{h_{m,n}},
\end{equation}
where 
\begin{equation}
{P_{{x_{n,k}}|\boldsymbol{\rm Y}}^{(t)}}({x_{n,k}}|\boldsymbol{\rm Y}) = \frac{1}{C_{x}} e^{(\Delta _{nk}^{\boldsymbol{\rm X}}(t ,{x_{n,k}}))}
\end{equation}
and
\begin{equation}
{P_{{h_{m,n}}|\boldsymbol{\rm Y}}^{(t)}}({h_{m,n}}|\boldsymbol{\rm Y}) = \frac{1}{C_{h}} e^{(\Delta _{mn}^{\boldsymbol{\rm H}}(t ,{h_{m,n}}))}.
\end{equation}
\(C_{x}\) and \(C_{h}\) are normalization parameters which ensure that the PDF integrates to one.

The existing BiGAMP algorithm is only applicable to scenarios with a small number of transmitting and receiving antennas. As the signal frequency rises and the number of antennas at the transmitter or receiver increases, the number of messages that need to be updated in the BiGAMP algorithm grows rapidly, leading to extremely high computational complexity and unaffordable processing time. For a message iteration, the number of messages requiring updates increases by \(2(M_{\rm t}-1)MK_{\rm d} + 2(N_{\rm t}-1)NK_{\rm d}\), when the number of transmitting antennas increases by a factor of \(M_{\rm t}\) and the number of receiving antennas increases by a factor of \(N_{\rm t}\) \cite{Parker2014,Parker2016}. How to reduce the computational complexity and processing time of DF-based estimation algorithms is crucial for practical implementation in mmWave massive MIMO systems.

\begin{figure}
\centering
\subfigure[Traditional PF-based method]{\includegraphics[width=\linewidth]{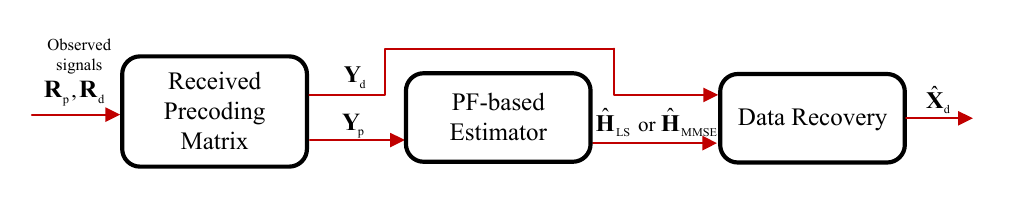}\label{fig:T_PF_Model}}
\subfigure[Original DF-based method]{\includegraphics[width=\linewidth]{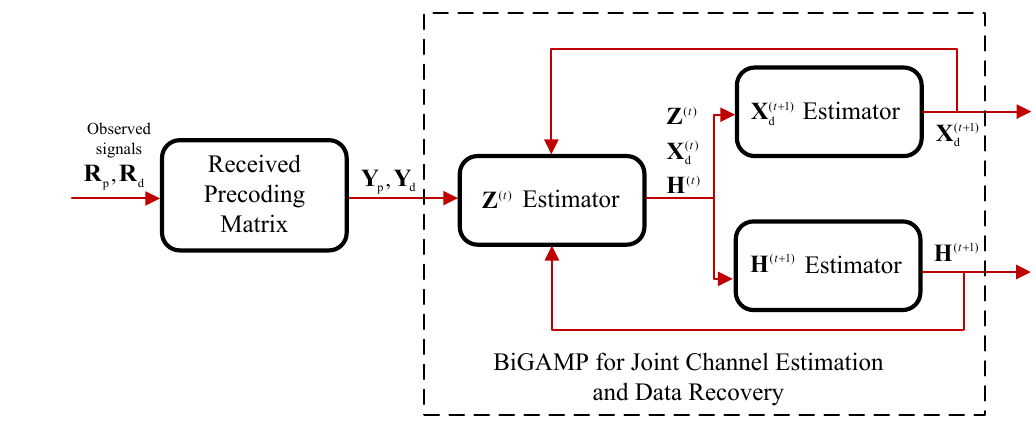}\label{fig:DF_Model}}
\subfigure[New PF-assisted JCD method]{\includegraphics[width=\linewidth]{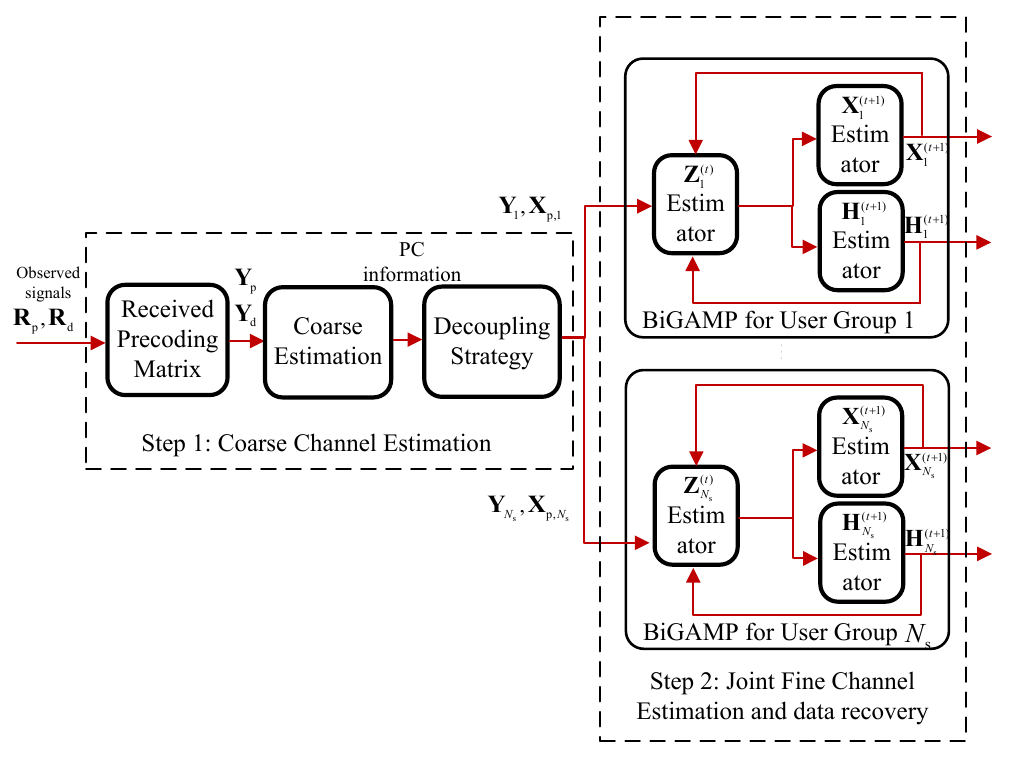}\label{fig:D_DF_Model}}
\caption{Block diagrams of the traditional PF-based method, original DF-based method and new PF-assisted JCD method.}\label{fig:Summary_PF_and_DF_method} 
\end{figure}

\begin{figure}
\centering
\includegraphics[width=\linewidth]{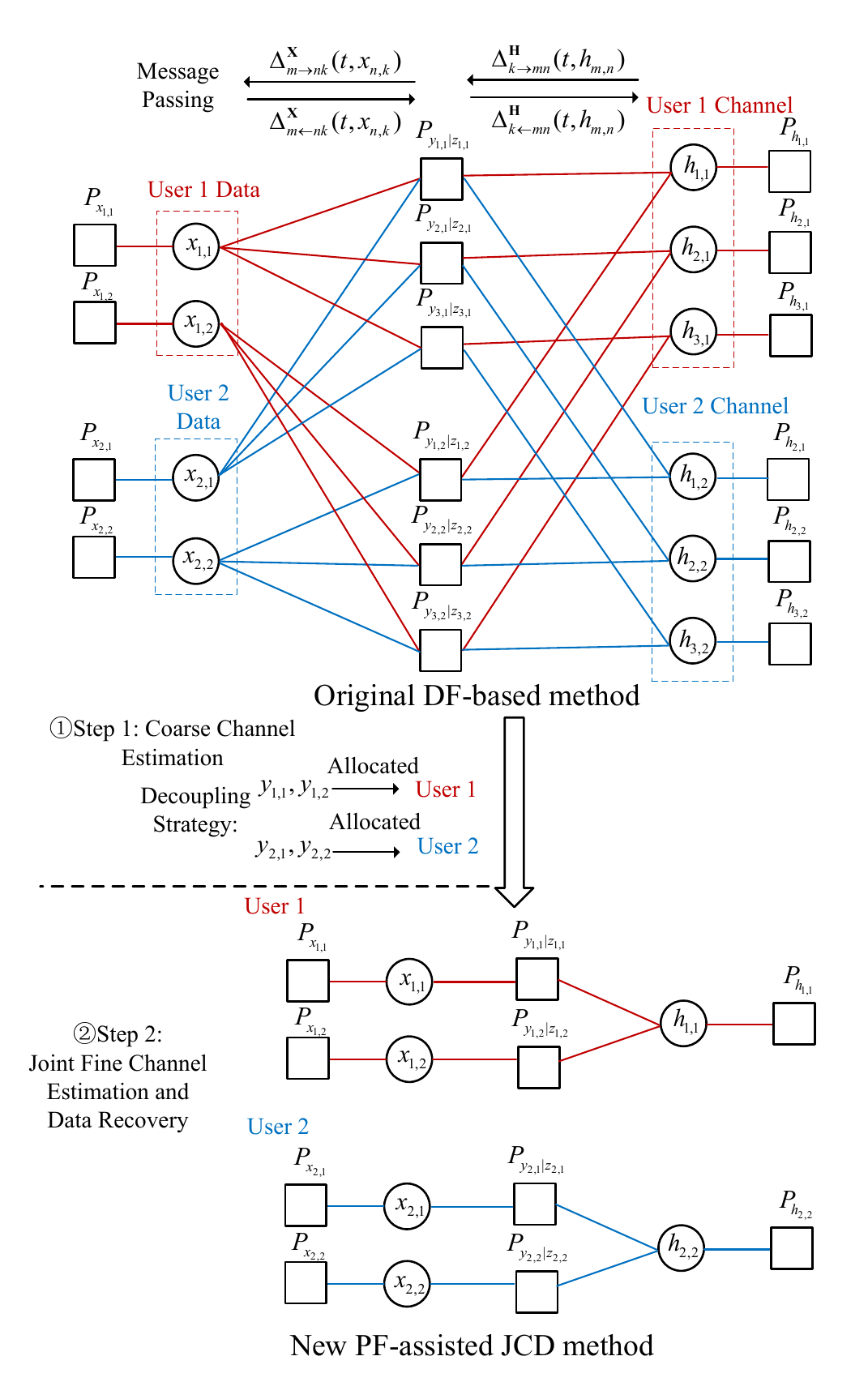}
\caption{Factor graphs of the original DF-based method and new PF-assisted JCD method with \(M=3, N=2, K=K_{\rm d}=2, M_{\rm track}=1, M_{\rm s}=0\).}\label{fig:Factor_Graph} 
\end{figure}

\section{PF-assisted JCD Method}\label{section_3}
In scenarios where the channel matrix \(\boldsymbol{\rm H}\) is sparse, the inefficiency of existing DF-based methods lies in the iterative updates of ineffective messages which do not contain information about uplink signals. In this section, the PF-assisted JCD method illustrated in Fig. \ref{fig:D_DF_Model} is proposed to solve the problem of inefficient message updates in DF-based methods using a two-stage estimation. In Stage 1, a coarse estimation method is proposed to preliminarily search for the PC of mmWave channels in a global angular domain through PF, and the designed decoupling strategy is applied to dimensionality reduction and multi-user decoupling of the multi-user estimation problem, which will be discussed in Section \ref{section_4_2}. In Stage 2, with the prior knowledge of PC obtained in Stage 1, a low computational complexity parallel algorithm is developed for joint fine channel estimation and data recovery utilizing both DF and PF, which will be discussed in Section \ref{section_4_3}.

\subsection{Stage 1: Coarse channel estimation in the global angular domain}\label{section_4_2}
when the CSI of mmWave channels is completely unknown, the purpose of coarse channel estimation in the global angular domain is to preliminarily obtain the information of multi-path components, thereby focusing on the PC of the channel matrix \(\boldsymbol{\rm H}\) to reduce the search range of AoA for PC in the subsequent estimation stage. At first, after being processed with the analog combiner \(\boldsymbol{\mathrm{F}}\), the observed PF signal \(\boldsymbol{\mathrm{Y}}_{\rm p}\) is given by
\begin{equation}
{{\boldsymbol{\mathrm{Y}}}_{\rm p}} = {{\boldsymbol{\mathrm{H}}}}{{\boldsymbol{\mathrm{X}}}_{\rm p}} + {{\boldsymbol{\mathrm{N}}}_{\rm p}},
\end{equation}
where the noise matrix \(\boldsymbol{\rm N}_{\rm p} = \boldsymbol{\mathrm{F}} {\boldsymbol{\mathrm{W}}}_{\rm p} \) is still Gaussian whose elements are independent and identically distributed (i.i.d.) and follow a complex Gaussian distribution \(\mathcal{CN}(0, \sigma_{\rm w}^2)\), because \( \boldsymbol{\mathrm{F}}\) is a unitary matrix.

The coarse channel estimation stage needs to address two issues: weak multi-path component detection in the global angular domain and the elimination of the interference caused by multi-user decoupling. For multi-path component detection, the misjudgments of multi-path components can result in a degradation of estimation performance and an increase in unnecessary AoA searches. For multi-user decoupling strategy, the overlapping multi-path components among different decoupling user groups will become interference. These interference severely degrades the estimation performance and is difficult to eliminate in the subsequent joint channel estimation and data recovery. The coarse estimation method is described in the following steps:

\subsubsection{Weak Multi-path Component Detection}
Firstly, the proposed coarse channel estimation method needs to obtain the equivalent channel matrix \(\boldsymbol{\rm H}\) for multi-path component detection in the global angular domain. Given the pilot symbols \({\boldsymbol{\mathrm{X}}}_{\rm p}\), the equivalent channel matrix \({\boldsymbol{\mathrm{H}}}\) can be directly estimated by LS method without requiring the probability statistics of channels as
\begin{equation}
{\hat{\boldsymbol{\mathrm{H}}}_{\rm LS}} = {{\boldsymbol{\mathrm{Y}}}_{\rm p}}{{\boldsymbol{\mathrm{X}}}_{\rm p}^{H}}({\boldsymbol{\mathrm{X}}}_{\rm p}{{\boldsymbol{\mathrm{X}}}_{\rm p}^{H}})^{-1}.
\end{equation}
where the pilot symbol matrix \(\boldsymbol{\rm X}_{\rm p}\) is designed to be row full rank (with \(K_{\rm p} \ge N\)).

To avoid the misjudgment of multi-path components caused by strong noise, the estimated channel matrix \({\hat{\boldsymbol{\mathrm{H}}}_{\rm LS}}\) requires the elimination of noise influence. Based on the spectral subtraction method, \({\hat{\boldsymbol{\mathrm{H}}}_{\rm LS}}\) is transformed into \(\left| {\hat{\boldsymbol{\mathrm{H}}}_{\rm LS}} \right|^2\) obtained by retaining only the magnitude squared of each element of matrix \({\hat{\boldsymbol{\mathrm{H}}}_{\rm LS}}\), and then the elements of matrix \(\left| {\hat{\boldsymbol{\mathrm{H}}}_{\rm LS}} \right|^2\) which only contain noise energy are set to zero. The threshold setting for multi-path component detection is based on the Neyman-Pearson (NP) criterion, and can be configured as
\begin{equation}
{\eta _{\rm NP}} = \Phi _{\rm RD}^{ - 1}(1 - \varepsilon ){\rm Var}\{\left| {{\boldsymbol{\rm N}_{{\rm{p}}}}} \right|_{i,j}^2{\rm{\} }} + {\rm E}\{ \left| {{\boldsymbol{\rm N}_{{\rm{p}}}}} \right|_{i,j}^2\}, 
\end{equation}
where \(\Phi _{\rm RD}^{ - 1}(\cdot)\) is the inverse function of the Rayleigh distribution's cumulative distribution function (CDF)
\begin{equation}
\begin{aligned}
{\Phi _{\rm RD}}(x) &= \int_{ 0 }^x {r \cdot {e^{ - \frac{{{r^2}}}{2}}}} {\rm d}r \\
&= 1 - e^{ - \frac{{{x^2}}}{2}}.  
\end{aligned}
\end{equation}
\( \varepsilon \) is the probability of noise-induced false alarm, \({\rm E}\{ \left| {{\boldsymbol{\rm N}_{{\rm{p}}}}} \right|_{i,j}^2\} \) and \({\rm Var}\{\left| {{\boldsymbol{\rm N}_{{\rm{p}}}}} \right|_{i,j}^2{\rm{\} }} \) represent the mean and variance of individual elements in matrix \(\left| {{\boldsymbol{\rm N}_{{\rm{p}}}}} \right|^2\), and their values are independent of \(i\) and \(j\). By eliminating the noise influence, the new channel matrix \( \tilde{\boldsymbol{\rm H}}_{\rm LS} \) can be rewritten as
\begin{equation}
\begin{aligned}
&\tilde{h}_{{\rm LS,}{i,j}} =
\left\{
             \begin{array}{ll}
             \left| {\hat{\boldsymbol{\rm H}}_{{\rm{LS}}}} \right|_{i,j}^2 - {\eta _{\rm NP}}, & \!\!\!  \left| {\hat{\boldsymbol{\rm H}}_{{\rm{LS}}}} \right|_{i,j}^2 > {\eta _{\rm NP}}  \\
             0, & \!\!\!  \left| {\hat{\boldsymbol{\rm H}}_{{\rm{LS}}}} \right|_{i,j}^2 \le {\eta _{\rm NP}} 
             \end{array},
\right.\\
&\quad\quad\quad\quad\quad\quad\quad\quad\quad\quad\quad\quad\quad\quad\quad\quad \forall i \in \{1,2,\ \ldots\ ,M_{\rm s}\}\\
&\quad\quad\quad\quad\quad\quad\quad\quad\quad\quad\quad\quad\quad\quad\quad\quad \forall j \in \{1,2,\ \ldots\ ,N\}\\
\end{aligned}
\end{equation}
where \({{{\tilde h}_{{\rm LS},i,j}}}\) is the element in the \(i\)th row and \(j\)th column of matrix \({\tilde{\boldsymbol{\mathrm{H}}}_{\rm LS}}\).

Then, the principal multi-path components are extracted from the denoised matrix \({\tilde{\boldsymbol{\mathrm{H}}}_{\rm LS}}\). The proposed PF-assisted JCD method tracks the \(M_{\rm track}\) strongest path for each user, and makes \(M_{\rm track} > L_{n}\). The angular index of the \(r\)th strongest pathes of the \(j\)th user can be given by
\begin{equation}\label{strategy_1}
\begin{aligned}
 {q_{r,j}} = \arg \mathop {\rm Find }\limits_i \left( {\left\| {{{\tilde h}_{{\rm LS},i,j}}} \right\|}, r  \right). \quad &\forall j \in \{ 1,2,\ ...\ ,N\}\\
 & \forall r \in \{ 1,2,\ ...\ ,M_{\rm track}\}
\end{aligned} 
\end{equation} 
where \(\arg \mathop {\rm Find }\limits_i (\cdot , r)\) is defined as a function which finds the index of the \(r\)th largest value from a set of indexed discrete values. Considering the diffusion of multi-path components, the search range for the \(r\)th multi-path component of the \(j\)th user can be set to the angular indices \(\{ {q_{r,j}}-\frac{M_{\rm s}}{2},\ ...\ , {q_{r,j}}+\frac{M_{\rm s}}{2} \}\). 

\subsubsection{The Decoupling Strategy of Multiple Users}
Using the processing results above, a multi-user decoupling strategy can be designed to achieve decoupling among multiple users as illustrated in Fig. \ref{fig:decoupling_strategy}. To avoid interference between different user groups, the design of the decoupling strategy needs to ensure that there is no overlapping multi-path components (i.e. no interference) among users in different user groups. Using the graph theory, the interference among users can be modeled as an undirected graph described by a tuple \(G = (V,I)\), where \(V=\{ v_{1},v_{2},\ ...\ , v_{N} \} \) denotes the set of vertices of the graph used to represent users, and \( I = \{ i_{1,2}, i_{1,3}, \ ...\ ,i_{2,3},\ ...\ ,i_{N,N-1} \} \) denotes the set of undirected edges connecting the nodes given by 
\begin{equation}
i _ {n_{1},{n}_{2}} = 
\left\{
             \begin{array}{ll}
             1, &  \text{Interference exists} \\
             0, &  \text{No Interference} \\
             \end{array}
\right..  \quad\quad  n_{1} \neq n_{2}
\end{equation}
When the edge weight of vertice \(v_{n_{1}}\) and vertice \(v_{n_{2}}\) are set to 1, it indicates that there is interference between them, while a weight of 0 indicates no interference.

In an undirected graph \(G\), if there is a path connecting vertice \(v_{n_{1}}\) and vertice \(v_{n_{2}}\), it indicates that user \(n_{1}\) and user \(n_{2}\) have a direct or indirect interference relationship, and hence can not be decoupled. The multi-user decoupling strategy assigns users with connected paths to the same user group, while users between different groups are unable to connect to each other. The search of users with connected paths can be accomplished by using breadth-first search (BFS) method or depth-first search (DFS) method. The decoupling user groups can be represented by \(\{ \mathcal{Q}_{n_{\rm s}}\}_{n_{\rm s}=1}^{N_{\rm s}}\), where \(N_{\rm s}\) is the number of decoupling user groups.

\begin{figure}
\centering
\includegraphics[width=\linewidth]{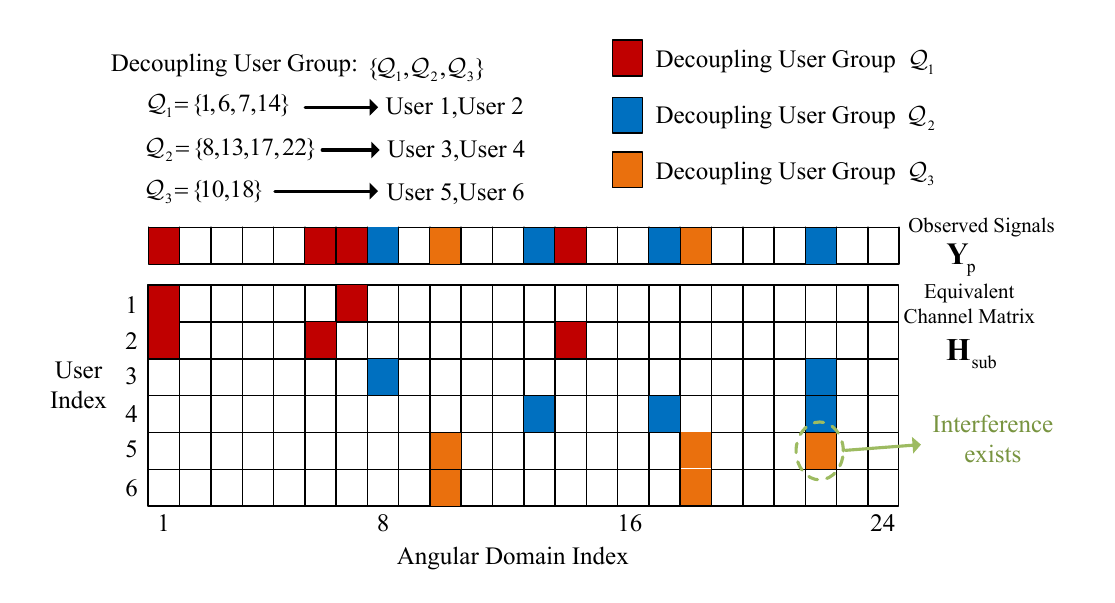}
\caption{Diagram of the multi-user decoupling strategy with \(M_{\rm track}=3, M_{\rm s} = 0\).}\label{fig:decoupling_strategy} 
\end{figure}

\subsection{Stage 2: Joint fine channel estimation and data recovery in the local angular domain}\label{section_4_3}
With the prior knowledge from the coarse channel estimation in the global angular domain, the receiver can eliminate ineffective searches within the angular domain where there are no multi-path components or only weak multi-path components exist, and narrow down the search range to the extent indicated by \( \{{q_{r,j}-\frac{M_{\rm s}}{2}},\ ...\ ,{q_{r,j}+\frac{M_{\rm s}}{2}} \}, \   \forall r \in \{ 1,2,...,M_{\rm track}\},\   \forall j \in \{ 1,2,...,N\}  \). The analog combiner \(\boldsymbol{\mathrm{F}}\) is set to choose the desired angular ranges, which is given by
\begin{equation}
\boldsymbol{\mathrm{F}} = \left[ \boldsymbol{\mathrm{F}}_{q_{1,1}}^{T}, \boldsymbol{\mathrm{F}}_{q_{2,1}}^{T},\ ...\ , \boldsymbol{\mathrm{F}}_{q_{{M_{\rm track}},{N}}}^{T} \right]^{T},
\end{equation}
where 
\begin{equation}
\begin{aligned}
\boldsymbol{\mathrm{F}}_{q_{r,j}} = \frac{1}{{\sqrt {{M}} }}\left[ {\boldsymbol{\mathrm{a}}_{{M}}}({q_{r,j}} -  \frac{M_{\rm s}}{2} ),\ ...\ ,{\boldsymbol{\mathrm{a}}_{{M}}}({q_{r,j}}),\right.\\
\left. \ ...\ ,{\boldsymbol{\mathrm{a}}_{{M}}}({q_{r,j}} +  \frac{M_{\rm s}}{2} ) \right]^{H}.  
\end{aligned}
\end{equation}
Because \(M_{\rm s}\cdot N \ll M \), processing through the analog beamforming matrix \(\boldsymbol{\rm F}\) enables the dimensionality reduction of the original high-dimensional signal. According to the decoupling user groups \(\{ \mathcal{Q}_{n_{\rm s}}\}_{n_{\rm s}=1}^{N_{\rm s}}\), the analog combiner \(\boldsymbol{\rm F}\) can rewritten as
\begin{equation}
\tilde{\boldsymbol{\mathrm{F}}} = \left[ \tilde{\boldsymbol{\mathrm{F}}}_{1}^{T}, \tilde{\boldsymbol{\mathrm{F}}}_{2}^{T},\ ...\ , \tilde{\boldsymbol{\mathrm{F}}}_{N_{\rm s}}^{T} \right]^{T},  
\end{equation}
where the analog combiner \(\tilde{\boldsymbol{\mathrm{F}}}_{n_{\rm s}}\) is composed of all submatrices \( \boldsymbol{\mathrm{F}}_{q_{r,j}} \) which are relevant to the decoupling user group \(\mathcal{Q}_{n_{\rm s}}\) (i.e. \(j \in \mathcal{Q}_{n_{\rm s}}\)). Note that the duplicated angular indices in the matrix \(\tilde{\boldsymbol{\mathrm{F}}}_{n_{\rm s}}\) need to be eliminated. Taking an example (\(M_{\rm track}=3\), \(M_{\rm s}=0\)) shown in Fig. \ref{fig:decoupling_strategy}, the analog beamforming matrix \(\tilde{\boldsymbol{\rm F}}\) can be represented as
\begin{align}
& \tilde{\boldsymbol{\mathrm{F}}} = \left[ \tilde{\boldsymbol{\mathrm{F}}}_{1}^{T}, \tilde{\boldsymbol{\mathrm{F}}}_{2}^{T},\tilde{\boldsymbol{\mathrm{F}}}_{3}^{T} \right]^{T}, \\
& \mathcal{Q}_{1} : \tilde{\boldsymbol{\mathrm{F}}}_{1}^{T} = \left[ \boldsymbol{\mathrm{F}}_{1}^{T}, \boldsymbol{\mathrm{F}}_{6}^{T}, \boldsymbol{\mathrm{F}}_{7}^{T},
\boldsymbol{\mathrm{F}}_{14}^{T}\right]\\
& \mathcal{Q}_{2} : \tilde{\boldsymbol{\mathrm{F}}}_{2}^{T} = \left[ \boldsymbol{\mathrm{F}}_{8}^{T}, \boldsymbol{\mathrm{F}}_{13}^{T}, \boldsymbol{\mathrm{F}}_{17}^{T},
\boldsymbol{\mathrm{F}}_{22}^{T}\right]\\
& \mathcal{Q}_{3} : \tilde{\boldsymbol{\mathrm{F}}}_{1}^{T} = \left[ \boldsymbol{\mathrm{F}}_{10}^{T}, \boldsymbol{\mathrm{F}}_{18}^{T}\right].
\end{align}

Then, the original high-dimensional bilinear problem can be transformed into multiple lower-dimensional subproblems. The original signal 
\eqref{bilinear_problem} can be decomposed into \(N_{\rm s}\) lower-dimensional signals, which is given by
\begin{equation}\label{full_original_signal_subproblem}
{{\boldsymbol{\mathrm{Y}}}_{n_{\rm s}}} = {{\boldsymbol{\mathrm{H}}}_{n_{\rm s}}}{{\boldsymbol{\mathrm{X}}}_{n_{\rm s}}} + {{\boldsymbol{\mathrm{N}}}_{n_{\rm s}}},   \quad \forall n_{\rm s} \in \{ 1,2,\ ...\ ,N_{\rm s}\} 
\end{equation}
where \({{\boldsymbol{\mathrm{H}}}_{n_{\rm s}}} = \tilde{\boldsymbol{\mathrm{F}}}_{n_{\rm s}} \boldsymbol{\mathrm{H}}\), \({\boldsymbol{\mathrm{N}}}_{n_{\rm s}} = \tilde{\boldsymbol{\mathrm{F}}}_{n_{\rm s}} \boldsymbol{\mathrm{N}} \) is a Gaussian random matrix whose elements follow i.i.d. complex Gaussian distribution \(\mathcal{CN}(0, \sigma^{2}_{\rm n} )\) (\(\sigma_{\rm n}^2= \sigma_{\rm w}^2\)), and \({{\boldsymbol{\mathrm{X}}}_{n_{\rm s}}} = [\boldsymbol{\rm X}_{{\rm p,}{n_{\rm s}}},\boldsymbol{\rm X}_{{\rm d,}{n_{\rm s}}}] \) is composed of symbol vectors of users which are in the set \(\mathcal{Q}_{n_{\rm s}}\). Then, the unknown parameters of the prior distribution can be estimated using the expectation maximization method, and the lower-dimensional bilinear subproblems can be parallelly processed by the decoupled EM-BiGAMP algorithm. The specific steps of the decoupled EM-BiGAMP algorithm for joint channel estimation and data recovery are as follows:

\subsubsection{\(\boldsymbol{\rm Z}^{(t)}\), \(\boldsymbol{\rm X}^{(t+1)}\) and \(\boldsymbol{\rm H}^{(t+1)}\) Estimation Step}
For \(n_{\rm s}\)th subproblem, \(\boldsymbol{\rm Z}^{(t)}\) is initially estimated. The elements of matrices \({{\boldsymbol{\mathrm{Y}}}_{n_{\rm s}}} \), \({{\boldsymbol{\mathrm{H}}}_{n_{\rm s}}} \), \({{\boldsymbol{\mathrm{X}}}_{{n_{\rm s}}}} \) and \({{\boldsymbol{\mathrm{Z}}}_{n_{\rm s}}} = {{\boldsymbol{\mathrm{H}}}_{n_{\rm s}}}{{\boldsymbol{\mathrm{X}}}_{{n_{\rm s}}}}\) can be respectively denoted as \(y_{m,k}\), \(h_{m,n}\), \(x_{n,k}\) and \(z_{m,k}\). For \(t\)th iteration, \(\bar p_{m,k}^{(t)}\) and \(\bar v_{m,k}^{p,(t)}\) are used to record the mean and variance of the prior distribution of \(z_{m,k}\) respectively. \(\hat p_{m,k}^{(t)}\) and \( v_{m,k}^{p,(t)}\) apply “Onsager” correction terms \cite{montanari2012}. \( \hat z_{m,k}^{(t)} \) and \(  v_{m,k}^{z,(t)} \) are the conditional mean and variance of the posterior PDF \(P_{z_{m,k}|y_{m,k},p_{m,k}}(z_{m,k}|y_{m,k},\hat p_{m,k}^{(t)},v_{m,k}^{p,(t)} ) = \frac{1}{C_{\rm z}} \cdot \mathcal{CN}({y_{m,k}};{z_{m,k}},{\sigma^2_{\rm n}}) \cdot \mathcal{CN}({z_{m,k}};\hat p_{m,k}^{(t)},v_{m,k}^{p,(t)}) \) (\(C_{\rm z}\) is a normalization parameter) for \(t\)th iteration. The mean and variance of the product of two Gaussians PDF can be obtained as
\begin{flalign}\label{algorithm_1_use}
& \hat z_{m,k}^{(t)} = \frac{{y_{m,k}}{v_{m,k}^{p,(t)}}+{\hat p_{m,k}^{(t)}}{\sigma_{\rm n}^2}}{v_{m,k}^{p,(t)} + \sigma_{\rm n}^2},\\
& v_{m,k}^{z,(t)} =  \frac{{v_{m,k}^{p,(t)}}{\sigma_{\rm n}^2}}{v_{m,k}^{p,(t)} + \sigma_{\rm n}^2}. \label{algorithm_2_use}
\end{flalign}

The messages flowing into the nodes to be estimated (i.e. \(x_{n,k}\) and \(h_{m,n}\)) can be linearly approximated using Taylor expansions, and its coefficients can be computed by \(\hat s_{m,k}^{(t)}\) and \(v_{m,n}^{s,(t)}\). For \(t\)th iteration, the posterior PDF of \(x_{m,k}\) and \(h_{m,n}\) given \(y_{m,k}\) can be given by \( P_{x_{n,k}|\boldsymbol{\rm Y}_{n_{\rm s}} ,r_{n,k}}(x_{n,k}|\boldsymbol{\rm Y}_{n_{\rm s}},\hat r_{n,k}^{(t)},v_{n,k}^{r,(t)} ) =  \frac{1}{C_{\rm x}}\cdot{P_{{x_{n,k}}}}({x_{n,k}})\cdot\mathcal{CN}({x_{n,k}};\hat r_{n,k}^{(t)},v_{n,k}^{r,(t)}) \) and \(P_{h_{m,n}|\boldsymbol{\rm Y}_{n_{\rm s}},q_{m,n}}(h_{m,n}|\boldsymbol{\rm Y}_{n_{\rm s}},\hat q_{m,n}^{(t)},v_{m,n}^{q,(t)} ) = \frac{1}{C_{\rm h}} \cdot {P_{h_{m,n}}}({h_{m,n}})\cdot\mathcal{CN}({h_{m,n}};\hat q_{m,n}^{(t)},v_{m,n}^{q,(t)}) \)(\(C_{\rm x}\) and \(C_{\rm h}\) are normalization parameters). Note that \(P_{x_{n,k}}(x_{n,k})\) and \(P_{h_{m,n}}(h_{m,n})\) are the prior PDF of \(x_{n,k}\) and \(h_{m,n}\). The PF symbols are known, and the DF symbols are encoded using Gaussian codebook, which is given by
\begin{equation}\label{X_G_Distribution}
\begin{aligned}
P_{x_{n,k}}(x_{n,k}) = 
\left\{
             \begin{array}{ll}
             \delta(x_{n,k} - \underline{x}_{n,k}), & \!\!\! \forall k \in \{1,2,\ ...\ ,K_{\rm p}\} \\
             \mathcal{CN}(x_{n,k} ;0, \sigma_{\rm x}^2), & \!\!\! \forall k \in \{K_{\rm p}+1,\ ...\ ,K\} \\
             \end{array}
\right.,
\end{aligned}
\end{equation}
where \(\{\underline{x}_{n,k}\}_{k=1}^{K_{\rm p}}\) are known PF symbols and \(\sigma_{\rm x}^2\) is the power of transmitted signals. Considering the sparsity of \(\boldsymbol{\mathrm{H}}_{n_{\rm s}}\), the elements of channel matrix \(\boldsymbol{\mathrm{H}}_{n_{\rm s}}\) can be modeled to follow independent Bernoulli-Gaussian (BG) distribution \cite{MoJianhua2018}
\begin{equation}\label{BG_Distribution}
P_{h_{m,n}}(h_{m,n}) = ( 1 - \lambda_{n} ){\delta(h_{m,n}) } + \lambda_{n} \cdot \mathcal{CN}( h_{m,n} ;0, \gamma_{n}),
\end{equation}
where \( \lambda_{n} \) represents the sparsity level of the \(n\)th user's channel vector (i.e. the \(n\)th column vector of the matrix \(\boldsymbol{\mathrm{H}}_{n_{\rm s}}\)), and \(\gamma_{n}\) is the variance of path gains. Note that parameters \( \lambda_{n} \) and \(\gamma_{n}\) are unknown and also require iterative estimation. Then, the posterior mean and variance of \(x_{n,k}\) and \(h_{m,n}\) can be computed as (\(\boldsymbol{\rm X}_{\rm p}\) is known and does not require estimation),
\begin{flalign}\label{algorithm_3_use}
& \hat x_{n,k}^{(t+1)} = \frac{{\hat r_{n,k}^{(t)}}{\sigma_{\rm x}^2}}{v_{n,k}^{r,(t)} + \sigma_{\rm x}^2}, \quad \forall k \in \{K_{\rm p}+1,\ ...\ ,K\}\\ 
& v_{n,k}^{x,(t+1)} =  \frac{{v_{n,k}^{r,(t)}}{\sigma_{\rm x}^2}}{v_{n,k}^{r,(t)} + \sigma_{\rm x}^2}, \quad \forall k \in \{K_{\rm p}+1,\ ...\ ,K\} \label{algorithm_4_use}
\end{flalign}
and
\begin{flalign}\label{algorithm_5_use}
& \hat h_{m,n}^{(t+1)} = \frac{{\alpha_{m,n}^{(t)}}{\hat q_{m,n}^{(t)}}{\gamma^{(t)}_n}}{v_{m,n}^{q,(t)} + \gamma^{(t)}_n},\\
& v_{m,n}^{h,(t+1)} = {\alpha_{m,n}^{(t)}}\left( \left| \frac{\hat h_{m,n}^{(t+1)}}{\alpha_{m,n}^{(t)}} \right|^2  + \frac{{\hat h_{m,n}^{(t+1)}}{v_{m,n}^{q,(t)}}}{{\alpha_{m,n}^{(t)}}{\hat q_{m,n}^{(t)}}}  \right) - \left| \hat h_{m,n}^{(t+1)}\right|^2 ,\label{algorithm_6_use}
\end{flalign}
where \(\alpha_{m,n}^{(t)}\) can be represented as
\begin{equation}
\alpha_{m,n}^{(t)} = \left(  1 + \frac{(1 - \lambda^{(t)}_{n} )\mathcal{CN}(0;\hat q_{m,n}^{(t)}, v_{m,n}^{q,(t)} )}{\lambda^{(t)}_{n} \cdot \mathcal{CN}(0;\hat q_{m,n}^{(t)}, \gamma_{n} + v_{m,n}^{q,(t)} ) } \right)^{-1}.
\end{equation}

\subsubsection{\({\lambda_{n} ^{(t + 1)}}\) and \({\gamma _n^{(t+1)}}\) Update Step}
The EM algorithm is employed to iteratively update parameters \(\lambda_{n}\) and \(\gamma_{n}\). For \(t\)th iteration, parameters \(\lambda_{n}^{(t+1)}\) and \(\gamma_{n}^{(t+1)}\) can be estimated as
\begin{equation}\label{mean_h_t}
\begin{aligned}
& {\lambda_{n} ^{(t + 1)}} = \arg \mathop {\max }\limits_{\lambda_{n}^{(t+1)}  \in [0,1]} \mathrm{E} \Big \{  \log P(\boldsymbol{\mathrm{Y}}_{{n_{\rm s}}},{{\boldsymbol{\mathrm{H}}}_{n_{\rm s}}},{{\boldsymbol{\mathrm{X}}}_{{n_{\rm s}}}}| \lambda^{(t+1)}_{n},\\ 
&   \{ \lambda_{\tilde{n}}^{(t)} \}_{\tilde{n}\neq n} ,\{ {\gamma _{\tilde{n}}^{(t)}}\}_{\tilde{n} = 0}^{{N_{{n_{\rm s}}}} - 1})
\left| {\boldsymbol{\mathrm{Y}}_{{n_{\rm s}}}, \{ \lambda_{\tilde{n}}^{(t)} \}_{\tilde{n} \neq n } ,\{ {\gamma_{\tilde{n}}^{(t)}}\} _{\tilde{n} = 0}^{{N_{{n_{\rm s}}}} - 1}} \right\},    
\end{aligned}
\end{equation}
\begin{equation}\label{var_h_t}
\begin{aligned}
& {\gamma_{n} ^{(t + 1)}} = \arg \!\!\!\! \mathop {\max }\limits_{\gamma_{n}^{(t+1)}  \in [0,+ \infty ]} \!\!\!\!\mathrm{E} \Big \{ \log P(\boldsymbol{\mathrm{Y}}_{{n_{\rm s}}},{{\boldsymbol{\mathrm{H}}}_{n_{\rm s}}},{{\boldsymbol{\mathrm{X}}}_{{n_{\rm s}}}}|\gamma_{n}^{(t+1)}, \\ 
&  \{ \lambda_{\tilde{n}}^{(t+1)}\}_{\tilde{n}=0}^{N_{n_{\rm s}}-1} ,{\{\gamma_{\tilde{n}}^{(t)}\}_{\tilde{n}\neq n}})
\left| {\boldsymbol{\mathrm{Y}}_{{n_{\rm s}}}, \{ \lambda_{\tilde{n}}^{(t+1)} \}_{\tilde{n}=0}^{N_{n_{\rm s}}-1} , {  \{ \gamma_{\tilde{n}}^{(t)} \}_{\tilde{n}\neq n }} } \right\}.    
\end{aligned}
\end{equation}
where \(N_{n_{\rm s}}\) is the number of users in the decoupling user group \(\mathcal{Q}_{n_{\rm s}}\). By eliminating some irrelevant random variables, the derivatives of equation \eqref{mean_h_t} and equation \eqref{var_h_t} are set to zero, we can compute the values of \({\lambda_{n} ^{(t + 1)}}\) and \({\gamma_{n} ^{(t + 1)}}\) as
\begin{equation}\label{algorithm_7_use}
\begin{aligned}
{\lambda_{n} ^{(t + 1)}} = \frac{1}{{{M_{{\rm r}, n_{\rm s}}}}}\sum\limits_{m = 0}^{M_{{\rm r}, n_{\rm s}} - 1} { {\alpha _{m,n}^{(t)}} }, 
\end{aligned}
\end{equation}
\begin{equation}\label{algorithm_8_use}
\begin{aligned}
{\gamma_{n} ^{(t + 1)}}  = \frac{1}{{\lambda_{n} ^{(t + 1)}}{M_{{\rm r}, n_{\rm s}}}}\sum\limits_{m = 0}^{M_{{\rm r}, n_{\rm s}} - 1} {\left( {v_{m,n}^{h,(t + 1)} + {{\left| {\hat h_{m,n}^{(t + 1)}} \right|}^2}} \right)} . 
\end{aligned}
\end{equation}
where \(M_{{\rm r}, n_{\rm s}}\) is the number of row vectors in matrix \(\boldsymbol{\rm H}_{n_{\rm s}}\).

\begin{table}
\centering
\caption{parameters for simulation analysis}\label{complexity_comparison}
\begin{tabular}{|| c | c | c||} 
\hline
\multirow{2}{*}{} & Original DF-based & PF-assisted JCD\\  
\multirow{2}{*}{} & Method \cite{Xiong2019,Wen2016} & Method \\
\hline\hline
\multirow{2}{*}{Multiplications
 } & \(10MN{K_{\rm d}}+9M{K_{\rm d}}\)  &   \(19  {M_{\rm r}}{K_{\rm d}}+16 {M_{\rm r}}\)  \\   
\multirow{2}{*}{} & \(+7N{K_{\rm d}}+16MN+N\)  &   \(+7{K_{\rm d}}+1\)  \\ \hline
\multirow{2}{*}{Additions
 } & \(10MN{K_{\rm d}}+6M{K_{\rm d}}\)  &   \(16 {M_{\rm r}}{K_{\rm d}}+ 8{M_{\rm r}}\)  \\   
\multirow{2}{*}{} & \(+4N{K_{\rm d}}+8MN\)  &  \(+4{K_{\rm d}}\)  \\ \hline
\end{tabular}
\label{TABLE1}
\end{table}

\subsection{Computational complexity}\label{section_4_4}
Compared with the joint estimation of \(\boldsymbol{\mathrm{H}}\) and \(\boldsymbol{\mathrm{X}}_{\rm d}\) using the original signal \eqref{original_signal}, the PF-assisted JCD method significantly reduces computational complexity. The whole procedure is provided in Algorithm \ref{ALGORITHM1}.

In Table \ref{complexity_comparison}, we assess the computational complexity of estimation methods in terms
of addition and multiplication operations for each iteration. Taking the original DF-based method \eqref{bilinear_problem} as an example, the computational complexity required for one iteration is approximately \(10MN{K_{\rm d}}+9M{K_{\rm d}}+7N{K_{\rm d}}+16MN+N\) multiplications and \(10MN{K_{\rm d}}+6M{K_{\rm d}}+4N{K_{\rm d}}+8MN\) additions. For the PF-assisted JCD method, assuming the original multi-path estimation problem can be decoupled into \(N\) single-user estimation subproblems and the search ranges of PC of different users do not overlap, the received uplink signal relevant to each user is reduced from \(M\) dimensions to \(M_{\rm r} = M_{\rm s}\cdot{M_{\rm track}}\) dimensions. For each subproblem, one iteration requires approximately \(19 {M_{\rm r}} {K_{\rm d}}+16{M_{\rm r}}+7{K_{\rm d}}+1\) multiplications and \(16 {M_{\rm r}} K_{\rm d}+ 8{M_{\rm r}} +4{K_{\rm d}}\) additions. All subproblems can be processed in parallel to reduce processing time. When \(M\) is large, the PF-assisted JCD method can effectively reduce the computational complexity and processing time.

\begin{algorithm}\label{ALGORITHM1}
 \SetAlgoNoLine 
 \caption{Decoupled EM-BIGAMP-based joint channel estimation and data recovery method}
  \KwIn{Observed signals \({{\boldsymbol{\mathrm{Y}}}_{{\rm d,}n_{\rm s}}}\) and PF symbol matrices \({\boldsymbol{\mathrm{X}}}_{{\rm p},{n_{\rm s}}},\ \forall {n_{\rm s}} \in \{ 1,2,\ ...\ ,N_{\rm s}\}\). }
  \KwOut{Estimated channel matrices \({\hat{\boldsymbol{\mathrm{H}}}}_{n_{\rm s}}\) and estimated DF symbol matrices \(\hat{\boldsymbol{\mathrm{X}}}_{{\rm d},{n_{\rm s}}}, \  \forall {n_{\rm s}} \in \{ 1,2,\ ...\ ,N_{\rm s}\} \). }
  \BlankLine
    \textbf{Initialize:}\\
    \(\forall m,n,k: z_{m,k}^{(0)} = 0, \hat{h}_{m,n}^{(1)} = 0, v_{m,n}^{h,(1)} =1, \hat{s}^{(0)}_{m,k}=0 , \lambda_{n}^{(1)} = 0.05, \sigma_{{\rm h},n}^{(1)}=0, t = 1 , \delta, T_{\rm max}\). \\
    \(\forall n, \forall k \in \{1,\ ...\ ,K_{\rm p}\}: x_{n,k}^{(1)}={\underline{x}_{n,k}}, v_{n,k}^{x,(1)} = 0 \). \\
    \(\forall n, \forall k \in \{K_{\rm p}+1,\ ...\ ,K\}: x_{n,k}^{(1)}=0, v_{n,k}^{x,(1)} = 1 \). \\
    \Repeat{${\left\| {{\hat{\boldsymbol{\mathrm{Z}}}_{{n_{\rm s}}}^{(t)}} - {{\hat{{\boldsymbol{\mathrm{Z}}}}}_{{n_{\rm s}}}^{(t-1)}}} \right\|} \le \delta  \left\| {{\hat{\boldsymbol{\mathrm{Z}}}_{{n_{\rm s}}}^{(t)}}} \right\| $ {\rm or} $\  t>T_{\rm max} $}
    {
       \(\forall n_{\rm s} :\) All subproblems can be processed in parallel.\\
       \textbf{Expectation phase:}\\
       \(\forall m,k :\) \(\bar v_{m,k}^{p,(t)} = {\sum\nolimits_n {\left| {\hat h_{m,n}^{(t)}} \right|} ^2}v_{n,k}^{x,(t)} + {\left| {\hat x_{n,k}^{(t)}} \right|^2}v_{m,n}^{h,(t)}\). \\
       \(\forall m,k :\) \(\bar p_{m,k}^{(t)} = \sum\nolimits_n {\hat h_{m,n}^{(t)}} \hat x_{n,k}^{(t)}\). \\
       \(\forall m,k :\) \(v_{m,k}^{p,(t)} = \bar v_{m,k}^{p,(t)} + \sum\nolimits_n {v_{m,n}^{h,(t)}v_{n,k}^{x,(t)}}. \) \\
       \(\forall m,k :\) \(\hat p_{m,k}^{(t)} = \bar p_{m,k}^{(t)} - \hat s_{m,k}^{(t - 1)}\bar v_{m,k}^{p,(t)}.  \)\\
       \(\forall m,k :\) Update \( \hat z_{m,k}^{(t)} \) and \(  v_{m,k}^{z,(t)} \) using \eqref{algorithm_1_use} and \eqref{algorithm_2_use}.\\
       \(\forall m,k :\) \(v_{m,k}^{s,(t)} = \frac{1}{{v_{m,k}^{p,(t)}}} - \frac{{v_{m,k}^{z,(t)}}}{{{{\left( {v_{m,k}^{p,(t)}} \right)}^2}}}\).\\
       \(\forall m,k :\) \(\hat s_{m,k}^{(t)} = \frac{{\hat z_{m,k}^{(t)} - \hat p_{m,k}^{(t)}}}{{v_{m,k}^{p,(t)}}}\).\\
       \(\forall n, \forall k \in \{K_{\rm p}+1,\ ...\ ,K \} :\)  \(v_{n,k}^{r,(t)} = {\left( {\sum\nolimits_m {{{\left| {\hat h_{m,n}^{(t)}} \right|}^2}v_{m,k}^{s,(t)}} } \right)^{ - 1}}\).\\
       \(\forall n, \forall k \in \{K_{\rm p}+1,\ ...\ ,K \} : \quad\quad\quad\quad\quad\quad\quad\)\(\hat r_{n,k}^{(t)} = \hat x_{n,k}^{(t)}(1 - v_{n,k}^{r,(t)}\sum\nolimits_m {v_{m,n}^{h,(t)}v_{m,k}^{s,(t)}} ) + v_{n,k}^{r,(t)}\sum\nolimits_m {{{(\hat h_{m,n}^{(t)})}^*}} \hat s_{m,k}^{(t)} \).\\
       \(\forall m,n :\) \(v_{m,n}^{q,(t)} = {\left( {\sum\nolimits_k {{{\left| {\hat x_{n,k}^{(t)}} \right|}^2}v_{m,k}^{s,(t)}} } \right)^{ - 1}}\). \\
       \(\forall m,n :\) \(\hat q_{m,n}^{(t)} = \hat h_{m,n}^{(t)}(1 - v_{m,n}^{q,(t)}\sum\nolimits_k {v_{n,k}^{x,(t)}v_{m,k}^{s,(t)}} ) + v_{m,n}^{q,(t)}\sum\nolimits_k {{{(\hat x_{n,k}^{(t)})}^*}\hat s_{m,k}^{(t)}} \). \\
       \(\forall n, \forall k \in \{1,\ ...\ ,K_{\rm p} \} :\)  Update \( \hat x_{n,k}^{(t+1)} \) and \( v_{n,k}^{x,(t+1)} \) by \( \hat x_{n,k}^{(t+1)} = \hat x_{n,k}^{(t)} \) and \( v_{n,k}^{x,(t+1)} = v_{n,k}^{x,(t)}\).\\
       \(\forall n, \forall k \in \{K_{\rm p}+1,\ ...\ ,K \} :\)  Update \( \hat x_{n,k}^{(t+1)} \) and \( v_{n,k}^{x,(t+1)} \) using \eqref{algorithm_3_use} and \eqref{algorithm_4_use}.\\
       \(\forall m,n :\) Update \( \hat h_{m,n}^{(t+1)} \) and \( v_{m,n}^{h,(t+1)} \) using \eqref{algorithm_5_use} and \eqref{algorithm_6_use}.\\
       \textbf{Maximization phase:}\\
       \(\forall n:\) Update \( \lambda_{n}^{(t+1)} \) using \eqref{algorithm_7_use}.\\
       \(\forall n:\) Update \( \gamma_{n}^{(t+1)} \) using \eqref{algorithm_8_use}.\\

       \textbf{Iteration index:} \\
       \(t= t+1\).\\

    }
    \bf{return} $ \hat{\boldsymbol{\mathrm{X}}}_{{\rm d},{n_{\rm s}}} $ and  $ \hat{\boldsymbol{\mathrm{H}}}_{n_{\rm s}},\   \forall {n_{\rm s}} \in \{ 1,2,\ ...\ ,N_{\rm s}\}. $\\

\end{algorithm}

\section{Performance analysis}\label{section_4}
In this section, we will analyze the performance of the PF-assisted JCD method and the impact of the two-stage estimation on estimation performance. First, consider the performance of joint channel estimation and data recovery for the bilinear problem \eqref{bilinear_problem}. In order to solve the bilinear problem \eqref{bilinear_problem}, the joint channel estimation and data recovery method aims to estimate the posterior means of \(\boldsymbol{\rm H}\) and \(\boldsymbol{\rm X}_{\rm d}\) using the knowledge of the observed signal \(\boldsymbol{\rm Y}\) and the prior PDF of \(\boldsymbol{\rm X}\) and \(\boldsymbol{\rm H}\) as
\begin{align}
\label{estimated_Xd_H_1} &\hat{\boldsymbol{\rm X}}_{\rm d} = \mathrm{E} \left\{   \boldsymbol{\rm X}_{\rm d} | \boldsymbol{\rm Y} \right\},\\
\label{estimated_Xd_H_2} &\hat{\boldsymbol{\rm H}} = \mathrm{E} \left\{   \boldsymbol{\rm H} | \boldsymbol{\rm Y} \right\}.
\end{align}
and mean square error (MSE) is used to evaluate the performance of the estimator, which can be given by
\begin{align}
\label{mse_Xd_H_1} &{\rm MSE}_{\boldsymbol{\rm X}_{\rm d}} = \mathrm{E} \left\{  \lVert \boldsymbol{\rm X}_{\rm d}  -  \hat{\boldsymbol{\rm X}}_{\rm d}  \rVert^{2}  \Big | \boldsymbol{\rm Y} \right\}, \\
\label{mse_Xd_H_2} &{\rm MSE}_{\boldsymbol{\rm H}} = \mathrm{E} \left\{  \lVert \boldsymbol{\rm H}  -  \hat{\boldsymbol{\rm H}}  \rVert^{2}  \Big | \boldsymbol{\rm Y} \right\}, 
\end{align}
Note that in the subsequent performance analysis, the elements of the symbols \(\boldsymbol{\rm X}_{\rm p}\) and \(\boldsymbol{\rm X}_{\rm d}\) follow the i.i.d. Gaussian distribution \eqref{X_G_Distribution}, and the elements of the equivalent channel \(\boldsymbol{\rm H}\) can be modeled to independently follow the BG distribution \eqref{BG_Distribution}.

Because of the coupling relationship between \(\boldsymbol{\rm H}\) and \(\boldsymbol{\rm X}\), it is infeasible to calculate \eqref{estimated_Xd_H_1} to \eqref{mse_Xd_H_2} directly. Inspired by \cite{Guo2005,Tanaka2002}, the asymptotic MSE performance of the joint estimation of \(\boldsymbol{\rm H}\) and \(\boldsymbol{\rm X}_{\rm d}\) can be obtained in a large system condition. A large system refers to a condition that the number of users tends to infinity, but the ratios of the number of users to other parameters remain fixed. Note that \(N \to \infty\), \(\frac{M}{N}=\alpha\), \(\frac{K_{\rm d}}{N}=\beta_{\rm d}\) and \(\frac{K_{\rm p}}{N}=\beta_{\rm p}\). With the help of the large deviation theory, the asymptotic MSE performance of the posterior mean estimators \(\hat{\boldsymbol{\rm X}}_{\rm d}\) and \(\hat{\boldsymbol{\rm H}}\) is obtained in Lemma \ref{Lemma 1}.

\begin{figure*}
\centering
\subfigure[Real channel]{\includegraphics[width=0.3\linewidth]{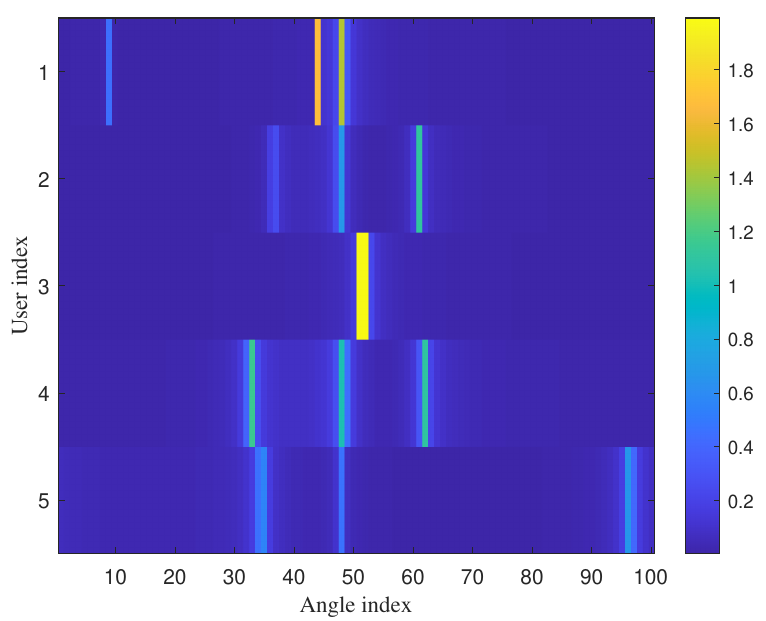}\label{fig:true_channel}}
\subfigure[Original DF-based method]{\includegraphics[width=0.3\linewidth]{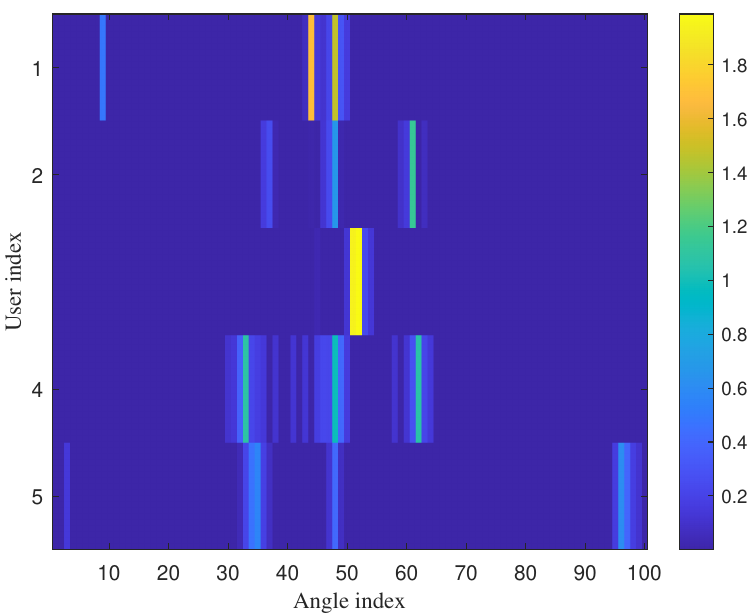}\label{fig:original_channel}}
\subfigure[PF-assisted JCD method (\(M_{\rm track}=1,M_{s}=10\))]{\includegraphics[width=0.3\linewidth]{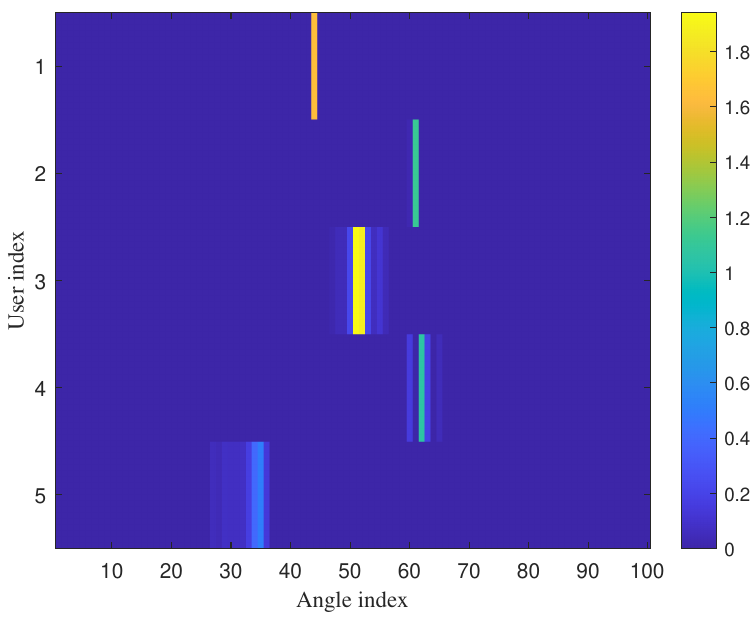}\label{fig:two-stage_channel}}
\caption{Comparison of the channel estimation of the original DF-based method, the PF-assisted JCD method and real channel with \(M=100\), \(N=5\) and \({\rm SNR}=0{\rm dB}\).}\label{fig:effect} 
\end{figure*}

Based on the results of Lemma \ref{Lemma 1}, considering a scenario where the number of antennas and symbols is much greater than the number of users (i.e. \(\alpha,\beta_{\rm d},\beta_{\rm p} \to \infty\)), the approximate expressions for \({\rm mse}_{{\boldsymbol{\rm X}}_{\rm d}}\) and \({\rm mse}_{{\boldsymbol{\rm H}}}\) can be obtained in Proposition \ref{Propostion 3}. Also, the effect of the two-stage estimation is discussed in Remark \ref{remark3}.

\begin{proposition}[The approximate expressions for \({\rm mse}_{{\boldsymbol{\rm X}}_{\rm d}}\) and \({\rm mse}_{{\boldsymbol{\rm H}}}\)]\label{Propostion 3}
A scenario where the communication system with a large number of antennas and symbols (and \(\alpha,\beta_{\rm d},\beta_{\rm p} \to \infty\)) is considered. With the knowledge of prior PDF of \(\boldsymbol{\rm X}\) and \(\boldsymbol{\rm H}\), using \eqref{Q_tilde_Q_1} to \eqref{Q_tilde_Q_6} in Lemma \ref{Lemma 1} and \eqref{x_d_mean_pro2} to \eqref{h_mse_pro2} in Lemma \ref{Lemma 2}, the MSE performance of \(\boldsymbol{\rm X}_{\rm d}\) and \(\boldsymbol{\rm H}\) can be approximated as
\begin{align}
\label{mse_x_d_pro3} & {\rm mse}_{\boldsymbol{\rm X}_{\rm d}}\approx \frac{\sigma_{\rm n}^2}{{\rm AE}_{\boldsymbol{\rm H}} - {\rm AMSE}_{\boldsymbol{\rm H}}  },\\
\label{mse_H_pro3} & {\rm mse}_{\boldsymbol{\rm H}} \approx \frac{\sigma_{\rm n}^2}{({\rm AE}_{\boldsymbol{\rm X}_{\rm d}} - {\rm AMSE}_{\boldsymbol{\rm X}_{\rm d}}) + {\rm AE}_{\boldsymbol{\rm X}_{\rm p}}  },
\end{align}
where \({\rm mse}_{\boldsymbol{\rm X}_{\rm d}}\) and \({\rm mse}_{\boldsymbol{\rm H}}\) represent the average MSE of individual elements of the matrices \(\boldsymbol{\rm X}_{\rm d}\) and \(\boldsymbol{\rm H}\), and \( \{ {\rm AE}_{\boldsymbol{\rm X}_{\rm p}},{\rm AE}_{\boldsymbol{\rm X}_{\rm d}},{\rm AE}_{\boldsymbol{\rm H}}, {\rm AMSE}_{\boldsymbol{\rm X}_{\rm d}} ,{\rm AMSE}_{\boldsymbol{\rm H}} \} \) are the key parameters affecting the performance of joint channel estimation and data recovery. Note that \({\rm AE}_{\boldsymbol{\rm H}} = \frac{{\rm E}\{\left\|  \boldsymbol{\rm H}  \right\|^2  \}}{N}\), \({\rm AE}_{\boldsymbol{\rm X}_{\rm d}} = \frac{{\rm E}\{\left\|  \boldsymbol{\rm X}_{\rm d} \right\|^2  \}}{N}\) and \({\rm AE}_{\boldsymbol{\rm X}_{\rm p}} = \frac{{\rm E}\{\left\|  \boldsymbol{\rm X}_{\rm p} \right\|^2  \}}{N}\) are the average energy (AE) of the channel matrix \(\boldsymbol{\rm H}\) and symbol matrix \(\boldsymbol{\rm X}_{\rm d}, \boldsymbol{\rm X}_{\rm p}\) per user. \({\rm AMSE}_{\boldsymbol{\rm H}} = \frac{{\rm E}\{\left\|  \boldsymbol{\rm H} - \hat{\boldsymbol{\rm H}}  \right\|^2  \}}{N} = M \cdot {\rm mse}_{\boldsymbol{\rm H}} \) and \({\rm AMSE}_{\boldsymbol{\rm X}_{\rm d}} = \frac{{\rm E}\{\left\|  \boldsymbol{\rm X}_{\rm d} - \hat{\boldsymbol{\rm X}}_{\rm d} \right\|^2  \}}{N} = {K_{\rm d}} \cdot {\rm mse}_{\boldsymbol{\rm X}_{\rm d}}\) are the average MSE (AMSE) of the channel matrix \(\boldsymbol{\rm H}\) and symbol matrix \( \boldsymbol{\rm X}_{\rm d}\) per user.
\end{proposition}

\begin{proof}
The proof is given in Appendix \ref{App_C}.
\end{proof}

\begin{remark}[The effect of two-stage estimation]\label{remark3}
From \eqref{mse_x_d_pro3} and \eqref{mse_H_pro3}, the key parameters that impact MSE performance are provided. It is noted that the two-stage estimation method primarily affects the performance of data recovery. Focusing on partial multi-path components and narrowing the search range for the AoA of multi-path components reduces \({\rm AE}_{\boldsymbol{\rm H}}\). To reduce the performance loss of data recovery, the primary principle of signal dimension reduction is to retain the AoA search range containing the maximum signal energy. The sparsity of the mmWave channel matrix \(\boldsymbol{\rm H}\) in the angular domain makes it feasible to utilize the PF-assisted JCD method for retaining the PC of channels, which contain a significant portion of the signal energy, through low-dimensional signals. 

Multi-user decoupling has a relatively minor impact on the performance of data recovery. When the difference in signal energy among all users is small, the decoupling of the original estimation problem does not lead to a significant change in \({\rm mse}_{\boldsymbol{\rm X}_{\rm d}}\). Thus, the proposed two-stage estimation method should decouple the original estimation problem into as many subproblems as possible to accelerate the solution of the multi-user estimation problem through parallel processing.    
\end{remark}

\begin{lemma}[Asymptotic MSE performance of the posterior mean estimators in the large system]\label{Lemma 1}
Under the condition of a large system, we can analogize the communication system to a thermodynamic system, and utilize the methods used to study thermodynamic systems to research communication systems \cite{Guo2005,Tanaka2002}. By employing the replica method and the large deviation theory, the asymptotic average free energy of this system can be given by
\begin{equation}\label{Free_energy_Pro_1}
\mathcal{F}{\rm{ = }} \mathop {\lim }\limits_{\gamma  \to 0} \frac{\partial }{{\partial \gamma }} \mathop {\rm Sadd} \limits_{\boldsymbol{\rm q}} \left\{ \mathcal{F}^{\gamma}(\boldsymbol{\rm q}) \right\},
\end{equation}
where \(\mathcal{F}^{\gamma}(\boldsymbol{\rm q})\) is the free energy obtained by the replica method, \(\mathop {\rm Sadd} \limits_{\boldsymbol{\rm x}}\{f(\boldsymbol{\rm x})\}\) represents the saddle point of \(f(\boldsymbol{\rm x})\) with respect to \(\boldsymbol{\rm x}\), and \(\boldsymbol{\rm q} = \{ q_{\boldsymbol{\rm X}_{\rm p}}, q_{\boldsymbol{\rm X}_{\rm d}}, q_{\boldsymbol{\rm H}}, \tilde{q}_{\boldsymbol{\rm X}_{\rm p}}, \tilde{q}_{\boldsymbol{\rm X}_{\rm d}}, \tilde{q}_{\boldsymbol{\rm H}} \}\) are the parameters which are associated with the MSE performance of estimated matrices \(\hat{\boldsymbol{\rm X}}_{\rm d}\) and \(\hat{\boldsymbol{\rm H}}\). Note that the parameters mentioned above are independent of \(m\), \(n\) and \(k\) under the large system condition. By finding the point of zero gradient of the free energy, the parameters \(\{ q_{\boldsymbol{\rm X}_{\rm p}}, q_{\boldsymbol{\rm X}_{\rm d}}, q_{\boldsymbol{\rm H}}, \tilde{q}_{\boldsymbol{\rm X}_{\rm p}}, \tilde{q}_{\boldsymbol{\rm X}_{\rm d}}, \tilde{q}_{\boldsymbol{\rm H}} \}\) can be obtained by solving the following equations
\begin{align}
\label{Q_tilde_Q_1} &\tilde{q}_{\boldsymbol{\rm H}} = \beta_{\rm d} q_{\boldsymbol{\rm X}_{\rm d}} {\chi }_{\rm d}+ \beta_{\rm p} q_{\boldsymbol{\rm X}_{\rm p}} {\chi }_{\rm p},\\
\label{Q_tilde_Q_2} &\tilde{q}_{\boldsymbol{\rm X}_{\rm d}} = \alpha q_{\boldsymbol{\rm H}} {\chi }_{\rm d},\\
\label{Q_tilde_Q_3} &\tilde{q}_{\boldsymbol{\rm X}_{\rm p}} = \alpha q_{\boldsymbol{\rm H}} {\chi }_{\rm p},\\
\label{Q_tilde_Q_4} &q_{\boldsymbol{\rm H}} = c_{\boldsymbol{\rm H}} - {\rm mse}_{\boldsymbol{\rm H}},\\
\label{Q_tilde_Q_5} &q_{\boldsymbol{\rm X}_{\rm d}} = c_{\boldsymbol{\rm X}_{\rm d}} - {\rm mse}_{\boldsymbol{\rm X}_{\rm d}},\\
\label{Q_tilde_Q_6} &q_{\boldsymbol{\rm X}_{\rm p}} = c_{\boldsymbol{\rm X}_{\rm p}},
\end{align}
where 
\begin{equation}
{\chi }_{\rm p} = \frac{1}{\frac{\sigma_{\rm n}^{2}}{N}+c_{\boldsymbol{\rm X}_{\rm p}} c_{\boldsymbol{\rm H}} - q_{\boldsymbol{\rm X}_{\rm p}} q_{\boldsymbol{\rm H}} },
\end{equation}
and 
\begin{equation}
{\chi }_{\rm d} = \frac{1}{\frac{\sigma_{\rm n}^{2}}{N}+c_{\boldsymbol{\rm X}_{\rm d}} c_{\boldsymbol{\rm H}} - q_{\boldsymbol{\rm X}_{\rm d}} q_{\boldsymbol{\rm H}} }.   
\end{equation}
Note that \(c_{\boldsymbol{\rm X}_{\rm p}} = c_{\boldsymbol{\rm X}_{\rm d}} = \frac{{\rm E}\{\|\boldsymbol{\rm X}\|^2\}}{NK}\) and \(c_{\boldsymbol{\rm H}} = \frac{{\rm E}\{\|\boldsymbol{\rm H}\|^2\}}{MN} \) are the average power of individual elements of the channel matrix \(\boldsymbol{\rm H}\) and symbol matrix \(\boldsymbol{\rm X}_{\rm d}\) respectively. The relationship between MSE performance \({\rm mse}_{\boldsymbol{\rm X}_{\rm d}}\), \({\rm mse}_{\boldsymbol{\rm H}}\) and parameters \(\{ q_{\boldsymbol{\rm X}_{\rm p}}, q_{\boldsymbol{\rm X}_{\rm d}}, q_{\boldsymbol{\rm H}}, \tilde{q}_{\boldsymbol{\rm X}_{\rm p}}, \tilde{q}_{\boldsymbol{\rm X}_{\rm d}}, \tilde{q}_{\boldsymbol{\rm H}} \}\) is provided in Lemma \ref{Lemma 2}.
\end{lemma}

\begin{proof}
The proof is given in Appendix \ref{App_A}.
\end{proof}

\begin{lemma}[The calculation of \({\rm mse}_{\boldsymbol{\rm X}_{\rm d}}\) and \({\rm mse}_{\boldsymbol{\rm H}}\)]\label{Lemma 2}
Directly computing \({\rm mse}_{\boldsymbol{\rm X}_{\rm d}}\) and \({\rm mse}_{\boldsymbol{\rm H}}\) is challenging, because the joint distribution of \(\{ \boldsymbol{\rm H},\boldsymbol{\rm X}_{\rm d}, \hat{\boldsymbol{\rm H}}, \hat{\boldsymbol{\rm X}}_{\rm d} \}\) is unavailable. With the help of the large system condition, the joint distribution of \(\{ \boldsymbol{\rm H},\boldsymbol{\rm X}_{\rm d}, \hat{\boldsymbol{\rm H}}, \hat{\boldsymbol{\rm X}}_{\rm d} \}\) can be decomposed into two independent parts: the joint distribution of \(\{ \boldsymbol{\rm H}, \hat{\boldsymbol{\rm H}} \}\) and the joint distribution of \(\{ \boldsymbol{\rm X}_{\rm d}, \hat{\boldsymbol{\rm X}}_{\rm d} \}\). For arbitrary elements \(x_{\rm d}\) and \(h\) of matrices \(\boldsymbol{\rm X}_{\rm d}\) and \(\boldsymbol{\rm H}\), the posterior estimations of \(x_{\rm d}\) and \(h\) based on signals \eqref{bilinear_problem} can be asymptotically approximated as posterior estimations in scalar Gaussian channels which can be modeled as 
\begin{align}
\label{X_d_AGMN_pro_2} & y_{x_{\rm d}} = \sqrt{\tilde{q}_{\boldsymbol{\rm X}_{\rm d}}} {x_{\rm d}} + w_{x_{\rm d}},\\
\label{H_AGMN_pro_2} & y_{h} = \sqrt{\tilde{q}_{\boldsymbol{\rm H}}} {h} + w_{h},
\end{align}
where \(w_{x_{\rm d}}\) and \(w_{x_{\rm d}}\) are the additive Gaussian noise with zero mean and unit variance. \( \sqrt{\tilde{q}_{\boldsymbol{\rm X}_{\rm d}} } \) and \( \sqrt{\tilde{q}_{\boldsymbol{\rm H}} } \) are the signal-to-noise ratios (SNR) of channels \eqref{X_d_AGMN_pro_2} and \eqref{H_AGMN_pro_2}. \(x_{\rm d}\) and \(h\) follow the original prior distributions \(P_{x_{\rm d}}\) and \(P_{h}\). Then, the posterior mean estimators of \(\hat{x}_{\rm d}\) and \(\hat{h}\) can be given by
\begin{align}
\label{x_d_mean_pro2} &\hat{x}_{\rm d} = \int{  {x}_{\rm d} \cdot P_{{x}_{\rm d}|y_{{x}_{\rm d}}}({x}_{\rm d}|y_{{x}_{\rm d}}) d{x}_{\rm d}},\\
\label{h_mean_pro2} &\hat{h} = \int{  {h} \cdot P_{{h}|y_{{h}}}({h}|y_{{h}}) d{h}},
\end{align}
where 
\begin{align}
\label{x_d_PP_pro2} & P_{{x}_{\rm d}|y_{{x}_{\rm d}}}({x}_{\rm d}|y_{{x}_{\rm d}}) = \frac{P_{y_{{x}_{\rm d}}|{x}_{\rm d}}(y_{{x}_{\rm d}}|{x}_{\rm d}) P_{{x}_{\rm d}}({x}_{\rm d}) }{ \int{P_{y_{{x}_{\rm d}}|{x}_{\rm d}}(y_{{x}_{\rm d}}|{x}_{\rm d}) P_{{x}_{\rm d}}({x}_{\rm d}){\rm d}{x_{\rm d}}}},\\
\label{h_PP_pro2} & P_{{h}|y_{{h}}}({h}|y_{{h}}) =  \frac{P_{y_{h}|h}(y_{h}|h) P_{h}(h) }{ \int{P_{y_{h}|h}(y_{h}|h) P_{h}(h){\rm d}{h}}},
\end{align}
are the posterior PDF of \(x_{\rm d}\) (and \(h\)) given \(y_{x_{\rm d}}\) (and \(y_{\rm h}\)). Finally, the asymptotic MSE of  \(x_{\rm d}\) and \(h\) can be given by
\begin{align}
\label{x_d_mse_pro2} & {\rm mse}_{{\boldsymbol{\rm X}}_{\rm d}} = \int{ (\hat{x}_{\rm d}-{x}_{\rm d})^2  \cdot P_{{x}_{\rm d}|y_{{x}_{\rm d}}}({x}_{\rm d}|y_{{x}_{\rm d}}) d{x}_{\rm d}},\\
\label{h_mse_pro2} & {\rm mse}_{{\boldsymbol{\rm H}}} = \int{  (\hat{h} - h ) \cdot P_{h|y_{h}}(h|y_{h}) dh},
\end{align}

\begin{proof}
The proof is given in Appendix \ref{App_B}.
\end{proof}

\end{lemma}


\section{Numerical results and discussion}\label{section_5}
In this section, the simulation results are provided for the proposed PF-assisted JCD method. The mmWave massive MIMO channels are modeled by the S-V model for simulation with 100 realizations \cite{Xiong2019}. The AoA of multi-path components are uniformly distributed within the interval \([0,  \pi ]\), the path gains of multi-path components of the \(n\)th user follow i.i.d Gaussian distribution \(\mathcal{CN}(0,1)\), and the number of paths of each user is set to 3 (i.e. \(L_1=L_2=...=L_{N_{\rm s}}=3\)). Each user ensures that the strongest multi-path component is not overwhelmed by noise (i.e. the SNR of the strongest multi-path components \(\ge 3{\rm dB}\)). The number of BS antennas \(M\) is 1000, the number of users \(N\) is 10 or 20, the transmitted signal contains 100 symbols (i.e. \(K=100\)), and \({\rm SNR} = 10 \log \frac{ \left\|  {\boldsymbol{\rm H}}{\boldsymbol{\rm X}}\right\|^2  }{\left\|  {\boldsymbol{\rm N}}\right\|^2}\). The parameters of the PF-assisted JCD method can be given by the number of tracked multi-path components \(M_{\rm track}=4\), the AoA search range of an individual multi-path component \(M_{\rm s} = 5\ {\rm or}\ 10\), and the false alarm probability \( \varepsilon = 10^{-5} \).

In Fig. \ref{Sim_Runtime_sub_1} and Fig. \ref{Sim_Runtime_sub_2}, we first compare the processing time of PF-based AMP method \cite{MoJianhua2018}, the original DF-based method \cite{Xiong2019,Wen2016} and the proposed PF-assisted JCD method under different conditions of \(M\), \(N\) and \(K_{\rm p}\). One sees that, compared with the original DF-based method and the PF-based AMP method, the PF-assisted JCD method significantly reduces processing time, especially when the number of antennas \(M\) and the number of users \(N\) are large. When \(M=1000\), \(N=20\), and \(K_{\rm p}=20\), the processing time of the PF-assisted JCD method is reduced by a factor of \(10^{2}\) compared to the original DF-based method. Moreover, the processing time of the PF-assisted JCD method is hardly affected by the increase in the number of antennas and users. It is because the strategy of narrowing the AoA search range to a finite interval \(M_{\rm track}\cdot M_{\rm s}\) restricts the processing time and computational complexity from increasing with the growth of the number of antennas \(M\), and the multi-user decoupling strategy reduces the processing time of the joint channel estimation and data recovery for \(N\) users to that of a single user (or a small number of users) through parallel processing. 

\begin{figure}
\centering
\includegraphics[width=\linewidth]{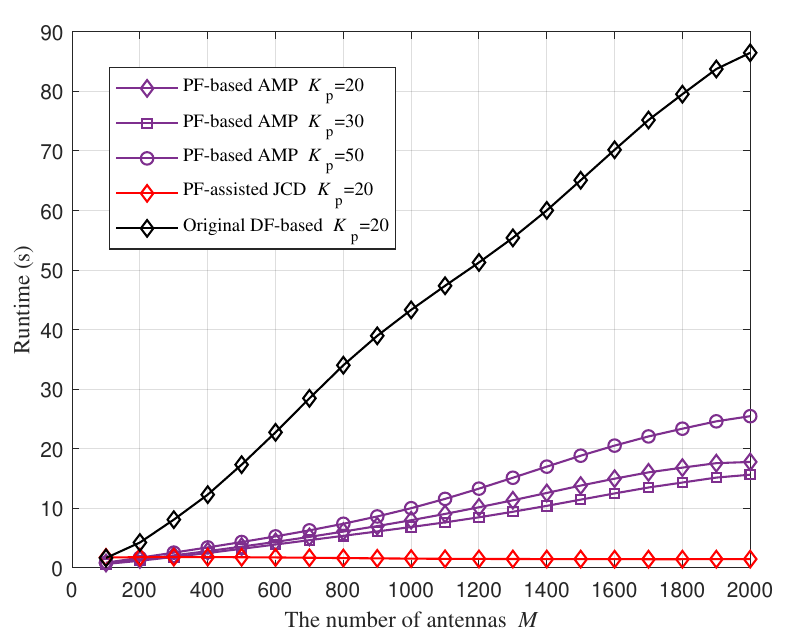}
\caption{Processing time of the PF-based AMP method, the original DF-based method and the PF-assisted JCD method at various \(M\) with \({\rm SNR} = 0 {\rm dB}\), \(\frac{M_{\rm s}}{M} = 0.01\), \(K = 100\) (note that \(K = K_{\rm d}+K_{\rm p}\)) and \(N=10\).}\label{Sim_Runtime_sub_1} 
\end{figure}

\begin{figure}
\centering
\includegraphics[width=\linewidth]{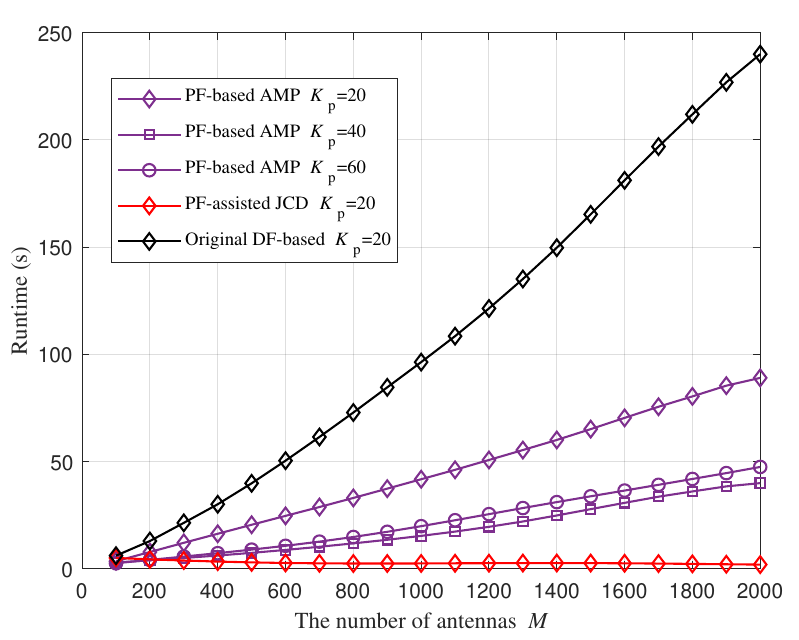}
\caption{Processing time of the PF-based AMP method, the original DF-based method and the PF-assisted JCD method at various \(M\) with \({\rm SNR} = 0 {\rm dB}\), \(\frac{M_{\rm s}}{M} = 0.005\), \(K = 100\) and \(N=20\).}\label{Sim_Runtime_sub_2} 
\end{figure}

\begin{figure}
\centering
\includegraphics[width=\linewidth]{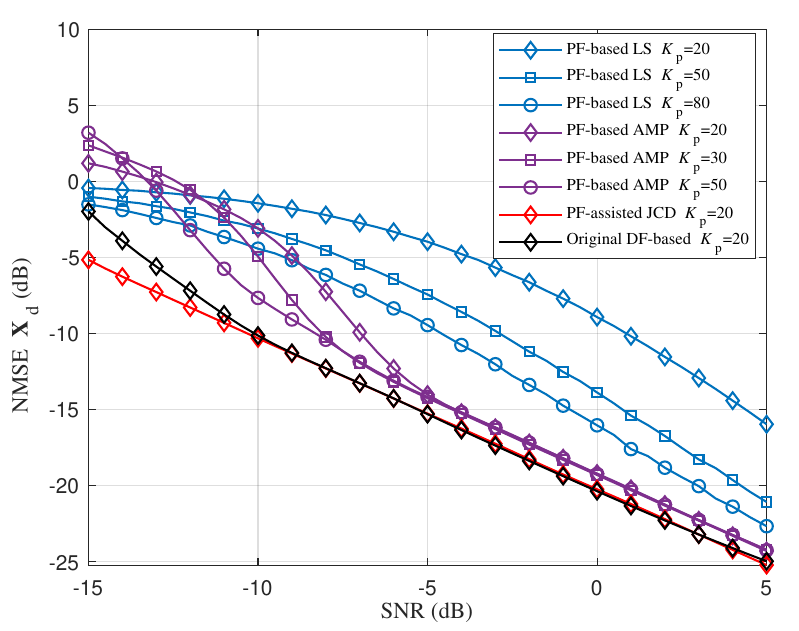}
\caption{NMSE \(\boldsymbol{{\rm X}}_{\rm d}\) performance of the traditional PF-based LS method (\(K_{\rm p} = 20,50,80\)), the PF-based AMP method (\(K_{\rm p}=20,30,50\)), the original DF-based method (\(K_{\rm p} = 20\)) and the PF-assisted JCD method (\(K_{\rm p} = 20\)) at various \(\rm{SNR}\) with \(N = 10\) and \(K = 100\).}\label{X_symbol_user=10} 
\end{figure}

\begin{figure}
\centering
\includegraphics[width=\linewidth]{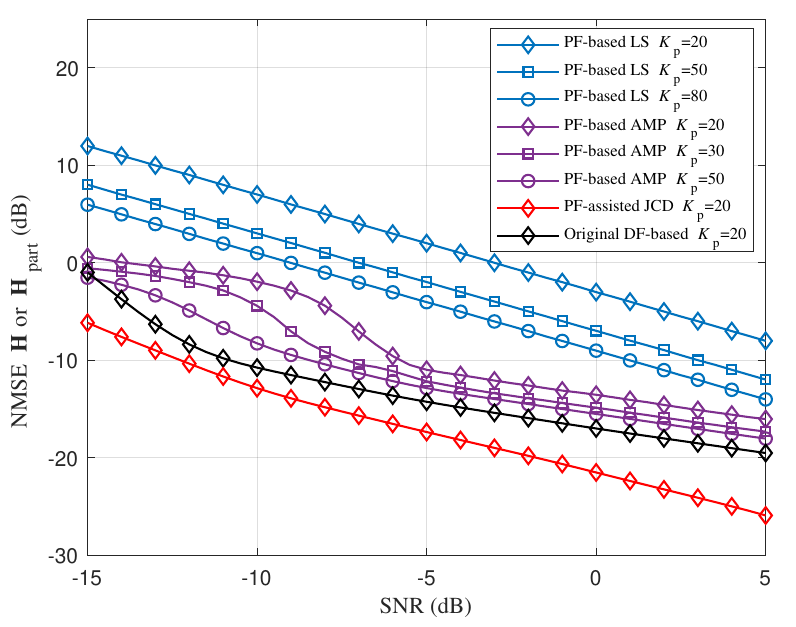}
\caption{NMSE \(\boldsymbol{{\rm H}}\) (or \(\boldsymbol{\rm H}_{\rm part}\)) performance of the traditional PF-based LS method (\(K_{\rm p} = 20,50,80\)), the PF-based AMP method (\(K_{\rm p}=20,30,50\)), the original DF-based method (\(K_{\rm p} = 20\)) and the PF-assisted JCD method (\(K_{\rm p} = 20\)) at various \(\rm{SNR}\) with \(N = 10\) and \(K = 100\).}\label{channel_err_user=10} 
\end{figure}

In Fig. \ref{X_symbol_user=10} to Fig. \ref{channel_err_user=20}, we analyze the performance of the traditional PF-based method (i.e. LS estimator), the PF-based AMP method, the original DF-based method and the PF-assisted JCD method. The normalized mean square error (NMSE) of \({\boldsymbol{\rm X}}_{\rm d}\) and \({\boldsymbol{\rm H}}\) is used to assess the performance of data recovery and channel estimation respectively, which can be given by 
\begin{equation}
{\rm NMSE}({\hat{\boldsymbol{\rm X}} }_{\rm d}) = \frac{{{{\left\| {{{\hat{\boldsymbol{\rm X}} }_{\rm d}}{\rm{ - }}{\boldsymbol{\rm X}_{\rm d}}} \right\|}}}}{{{{\left\| {{\boldsymbol{\rm X}_{\rm d}}} \right\|}}}},
\end{equation}
\begin{equation}
{\rm NMSE}({\hat{\boldsymbol{\rm H}} }) = \frac{{{{\left\| {{{\hat{\boldsymbol{\rm H}} }}{\rm{ - }}{\boldsymbol{\rm H}}} \right\|}}}}{{{{\left\| {{\boldsymbol{\rm H}}} \right\|}}}}.
\end{equation}

Compared to the traditional PF-based methods, the PF-assisted JCD method and the original DF-based method utilize DF for joint channel estimation and data recovery, thereby reducing a significant amount of PF overhead. In simulation, the PF-assisted JCD method and the original DF-based method demonstrate a notable performance gain over the traditional PF-based method in both channel estimation and data recovery. Compared to the original DF-based method, the PF-assisted JCD method reduces the dimensionality of the observed signals while preserving the information of the principal components of channels. The simulation results shows that the PF-assisted JCD method can achieve the same performance of data recovery as the the original DF-based method. Even in a low SNR regime, the PF-assisted JCD method exhibits greater robustness compared to the original DF-based method, owing to the denoising function obtained by the strategy of eliminating invalid search ranges. In summary, the PF-assisted JCD method is an excellent joint channel estimation and data recovery method, which offers an implementation solution with limited resources (including PF overhead, computational complexity and processing time) while achieving better performances than the original DF-based method in most cases.


\begin{figure}
\centering
\includegraphics[width=\linewidth]{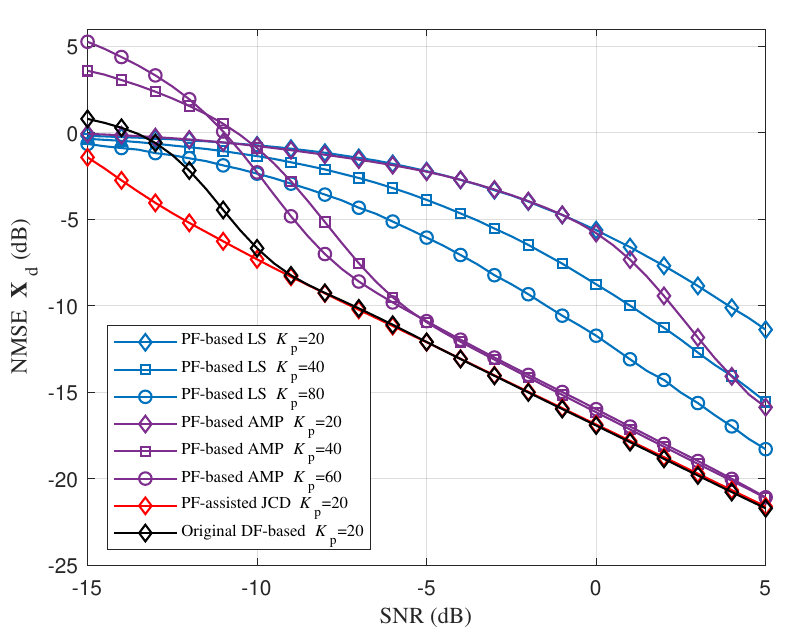}
\caption{NMSE \(\boldsymbol{{\rm X}}_{\rm d}\) performance of the traditional PF-based LS method (\(K_{\rm p} = 20,40,80\)), the PF-based AMP method (\(K_{\rm p}=20,40,60\)), the original DF-based method (\(K_{\rm p} = 20\)) and the PF-assisted JCD method (\(K_{\rm p} = 20\)) at various \(\rm{SNR}\) with \(N = 20\) and \(K = 100\).}\label{X_symbol_user=20} 
\end{figure}

\begin{figure}
\centering
\includegraphics[width=\linewidth]{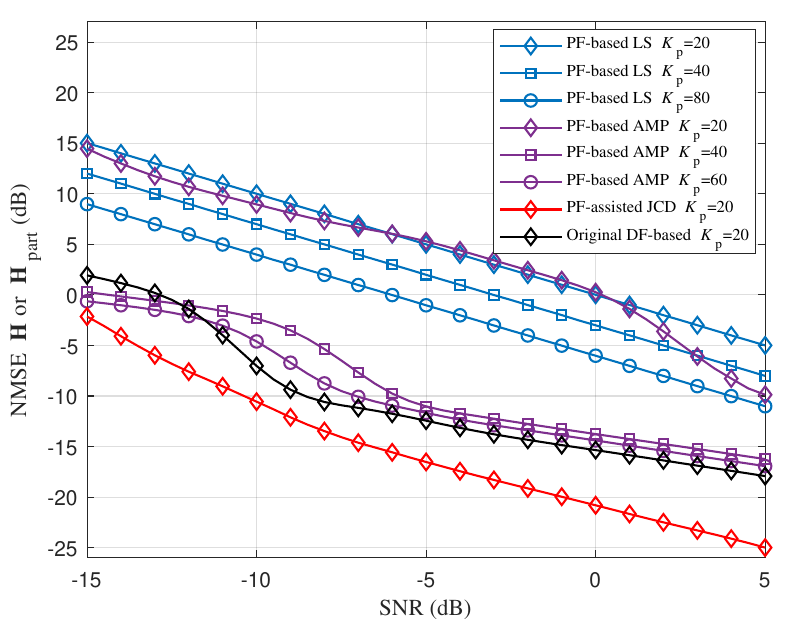}
\caption{NMSE \(\boldsymbol{{\rm H}}\) (or \(\boldsymbol{\rm H}_{\rm part}\)) performance of the traditional PF-based LS method (\(K_{\rm p} = 20,40,80\)), the PF-based AMP method (\(K_{\rm p}=20,40,60\)), the original DF-based method (\(K_{\rm p} = 20\)) and the PF-assisted JCD method (\(K_{\rm p} = 20\)) at various \(\rm{SNR}\) with \(N = 20\) and \(K = 100\).}\label{channel_err_user=20} 
\end{figure}

\section{Conclusion}\label{section_6}
In this paper, we have proposed a PF-assisted JCD method which achieved low computational complexity and low PF overhead. By combining the advantages of PF-based and DF-based methods, the proposed PF-assisted JCD method has initially utilized PF to acquire the AoA information of PC of channels, and then DF to achieve the joint estimation of the channel matrix \(\boldsymbol{\rm H}\) and symbol matrix \(\boldsymbol{\rm X}_{\rm d}\). The theoretical analysis has revealed that, in terms of estimation performance, the PF-assisted JCD method was proven not to cause performance loss in joint channel estimation and data recovery. In terms of computational complexity, the PF-assisted JCD method achieved a reduction of approximately \(\alpha_{\rm r} \approx \frac{M}{M_{\rm track}\cdot M_{\rm s}} \) times compared to the original DF-based method. The simulation results have shown that the PF-assisted JCD method performed comparably to the original DF-based method and even exhibited greater robustness in a low SNR regime. The processing time of the PF-assisted JCD method was approximately \(10^2\) times less than that of the original DF-based method.

\appendices
\section{The Proof of Lemma 1}\label{App_A}
For a many-body thermodynamic system, the microscopic state can be described by the matrices \(\boldsymbol{\rm X}\) and \(\boldsymbol{\rm H}\). The asymptotic average free energy of this system can be written as
\begin{equation}\label{free_energy_1}
\mathcal{F}{\rm{ = }}\mathop {{\rm{lim}}}\limits_{N \to \infty }\frac{{\rm{1}}}{{{N^{\rm{2}}}}}{{\rm E}_{\boldsymbol{\rm Y}}}{\rm{\{ }}{\log{P}_{\boldsymbol{\rm Y}}}{\rm{(}}{\boldsymbol{\rm Y}}{\rm{)\} }},
\end{equation}
where 
\begin{equation}
{{P}_{\boldsymbol{\rm Y}}}(\boldsymbol{\rm Y}) = \iint {{{P}_{\boldsymbol{\rm Y}|{\boldsymbol{\rm H}},\boldsymbol{\rm X}}}{\rm{(}}\boldsymbol{\rm Y}|\boldsymbol{\rm X},\boldsymbol{\rm H}{){P}_{\boldsymbol{\rm X}}(\boldsymbol{\rm X}){P}_{\boldsymbol{\rm H}}(\boldsymbol{\rm H}){\rm d}\boldsymbol{\rm X}{\rm d}\boldsymbol{\rm H}}}
\end{equation}
is the partition function. The expectation \({\rm E}_{\boldsymbol{\rm Y}}(\cdot)\) is taken over observed signals \(\boldsymbol{\rm Y}\) with the distribution \({{P}_{\boldsymbol{\rm Y}}}(\boldsymbol{\rm Y})\). 

The direct calculation of the free energy \eqref{free_energy_1} is hard. Using the replica method originally used in statistical physics, the free energy \eqref{free_energy_1} can be reformulated as\footnote{The calculation of the equation \eqref{free_energy_RM} is referenced from 
\begin{equation}
\mathop {\lim }\limits_{\gamma  \to 0} \frac{\partial }{{\partial \gamma }}{\rm E}{\rm{\{ }}{a^\gamma }{\rm{\}  = }}\mathop {\lim }\limits_{\gamma  \to 0} \frac{{{\rm E}\{ {a^\gamma }\log a\} }}{{{\rm E}\{ {a^\gamma }\} }} = {\rm E}\{ \log a\},
\end{equation}
where \(\gamma > 0 \).} 
\begin{equation}\label{free_energy_RM}
\mathcal{F}{\rm{ = }}\mathop {{\rm{lim}}}\limits_{N \to \infty } \frac{{\rm{1}}}{{{N^{\rm{2}}}}}\mathop {\lim }\limits_{\gamma  \to 0} \frac{\partial }{{\partial \gamma }}\log{{\rm E}_{\boldsymbol{\rm Y}}}{\{ P}_{\boldsymbol{\rm Y}}^\gamma {\rm{(}}{\boldsymbol{\rm Y}}{\rm{)\} }},
\end{equation}
where 
\begin{equation}\label{app_ey}
\begin{aligned}
& {{\rm E}_{\boldsymbol{\rm Y}}}{\{ P}_{\boldsymbol{\rm Y}}^\gamma {\rm{(}}{\boldsymbol{\rm Y}}{\rm{)\} }} =  \\
& \int {\prod\limits_{a = 0}^\gamma  {\left\{ {{\rm d}{\boldsymbol{\rm H}^{(a)}}{\rm d}{\boldsymbol{\rm X}^{(a)}} \cdot {P_{\boldsymbol{\rm Y}|\boldsymbol{\rm Z}}}(\boldsymbol{\rm Y}|{\boldsymbol{\rm Z}^{(a)}} = {\boldsymbol{\rm H}^{(a)}}{\boldsymbol{\rm X}^{(a)}}) } \right\} }}\\
& {{  \cdot \left\{ {\rm d}{\boldsymbol{\rm Y}} \cdot {P_{\boldsymbol{\rm H},\boldsymbol{\rm X}}}(\{ {\boldsymbol{\rm H}^{(a)}}\} _{a = 0}^\gamma ,\{ {\boldsymbol{\rm X}^{(a)}}\} _{a = 0}^\gamma ) \right\} } },    
\end{aligned}
\end{equation}
where \( \boldsymbol{\rm X}^{(0)} \) and \( \boldsymbol{\rm H}^{(0)} \) are the original channel matrix and symbol matrix whose elements independently follow the distribution \(P_{\boldsymbol{\rm X}}(\boldsymbol{\rm X})\) \eqref{X_G_Distribution} and \(P_{\boldsymbol{\rm H}}(\boldsymbol{\rm H})\) \eqref{BG_Distribution}. For simplicity, we assume that each user has the same parameters \(\lambda_{n}=\lambda\) and \(\gamma_{n}=\sigma_{{\rm h}}^2,\quad\forall n\in \{1,2,\ ...\ ,N\} \) in the BG distribution. The generalized version can be obtained by decomposing matrix \(\boldsymbol{\rm H}\) into column vectors \(\{\boldsymbol{\rm h}_{n}\}_{n=1}^{N}\), but this becomes too cumbersome. \( \{\boldsymbol{\rm X}^{(a)} \}_{a=1}^{\gamma} \) and \( \{ \boldsymbol{\rm H}^{(a)} \}_{a=1}^{\gamma} \) are the replicas of \( \boldsymbol{\rm X}^{(0)} \) and \( \boldsymbol{\rm H}^{(0)} \) generated by the retrochannels \cite{Guo2005}. The new matrices \(\boldsymbol{\rm Q}_{\boldsymbol{\rm H}} \in {{\mathbb{C}}^{(\gamma+1) \times (\gamma+1)}} \) and \( \{\boldsymbol{\rm Q}_{\boldsymbol{\rm X}_{v}}\} \in {{\mathbb{C}}^{(\gamma+1) \times (\gamma+1)}}, \forall v \in \{{\rm p},{\rm d}\} \) are introduced to describe the microstates \( \{\boldsymbol{\rm X}_{\rm p}^{(a)} \}_{a=0}^{\gamma} \), \( \{\boldsymbol{\rm X}_{\rm d}^{(a)} \}_{a=0}^{\gamma} \) and \( \{ \boldsymbol{\rm H}^{(a)} \}_{a=0}^{\gamma} \), which can be given by
\begin{equation}\label{app_qh}
q_{\boldsymbol{\rm H},a,b} = \frac{1}{{MN}}  \sum\limits_{m = 1}^M \sum\limits_{n = 1}^N {h_{m,n}^{(a)} 
\cdot {h_{m,n}^{(b)}}} ,
\end{equation}
\begin{equation}\label{app_qx}
q_{\boldsymbol{\rm X}_{\rm p},a,b} = \frac{1}{{NK_{\rm p}}}  \sum\limits_{n = 1}^N \sum\limits_{k = 1}^{K_{\rm p}} {x_{n,k}^{(a)} 
\cdot {x_{n,k}^{(b)}}} ,
\end{equation}
\begin{equation}\label{app_qxd}
q_{\boldsymbol{\rm X}_{\rm d},a,b} = \frac{1}{{NK_{\rm d}}}  \sum\limits_{n = 1}^N \sum\limits_{k = K_{\rm p}+1}^K {x_{n,k}^{(a)} 
\cdot {x_{n,k}^{(b)}}} .
\end{equation}
Using the central limit theorem, \(  \{ \boldsymbol{\rm Z}_{v}^{(a)}\}_{a=0}^{\gamma},\ \forall v \in \{{\rm p},{\rm d}\}\) can be modeled as multivariate Gaussian random variables
whose distribution is given by
\begin{equation}\label{app_qz}
\begin{aligned}
& P_{\boldsymbol{\rm Z}_{v}}(\{ {\boldsymbol{\rm Z}_{v}^{(a)}}\} _{a = 0}^\gamma|\boldsymbol{\rm Q}_{\boldsymbol{\rm H}},\boldsymbol{\rm Q}_{\boldsymbol{\rm X}_{v}}) = \\
&\prod\limits_{m,k} {\frac{1}{{\sqrt {{{(2\pi )}^{(\gamma  + 1)}}\det (\boldsymbol{\Sigma}_{v} )} }} \cdot } {e^{\left( { - \frac{1}{2}  {\boldsymbol{\rm z}_{m,k}^H 
(\boldsymbol{\Sigma}_{v}^{-1}) \boldsymbol{\rm z}_{m,k}} } \right)}},\quad\forall v \in \{{\rm p},{\rm d}\}
\end{aligned}
\end{equation}
where \( \boldsymbol{\rm z}_{m,k} = [z_{m,k}^{(0)}, z_{m,k}^{(1)}, \ ...\ ,z_{m,k}^{(\gamma)}] \) is the random vector composed of all replicas of \(z_{m,k}\), \( (\boldsymbol{\Sigma}_{v} )_{a,b}= q_{\boldsymbol{\rm H},a,b} \cdot q_{\boldsymbol{\rm X}_{v},a,b},\ \forall v \in \{{\rm p},{\rm d}\}\) is the covariance matrix of the random vector \(\boldsymbol{\rm z}_{m,k}\). 

From the above \eqref{app_ey} to \eqref{app_qz}, \({{\rm E}_{\boldsymbol{\rm Y}}}{\{ P}_{\boldsymbol{\rm Y}}^\gamma {\rm{(}}{\boldsymbol{\rm Y}}{\rm{)\} }}\) can be calculated as
\begin{equation}
\begin{aligned}
& {{\rm E}_{\boldsymbol{\rm Y}}}{\{ P}_{\boldsymbol{\rm Y}}^\gamma {\rm{(}}{\boldsymbol{\rm Y}}{\rm{)\} }} = \\ 
& \int {\rm d}{\boldsymbol{\rm Y}} \cdot \prod\limits_{a = 0}^\gamma  \bigg\{ { {P_{\boldsymbol{\rm Y}|\boldsymbol{\rm Z}}}(\boldsymbol{\rm Y}|{\boldsymbol{\rm Z}_{\rm p}^{(a)}},{\boldsymbol{\rm Z}_{\rm d}^{(a)}} )} \\
&\cdot \prod\limits_{v\in \{ {\rm p},{\rm d}\} } \left\{ {\rm d}{{\boldsymbol{\rm Z}}_{v}^{(a)}} \cdot P_{\boldsymbol{\rm Z}_{v}}(\{ {\boldsymbol{\rm Z}_{v}^{(\tilde{a})}}\} _{\tilde{a} = 0}^\gamma|\boldsymbol{\rm Q}_{\boldsymbol{\rm H}},\boldsymbol{\rm Q}_{\boldsymbol{\rm X}_{v}}) \right\} \bigg\} \\ 
&\cdot {\left\{ {{\rm d}{\boldsymbol{\rm Q}}_{\boldsymbol{\rm H}} \cdot {P_{{{\boldsymbol{\rm Q}}_{\boldsymbol{\rm H}}}}}({{\boldsymbol{\rm Q}}_{\boldsymbol{\rm H}}})} \right\}} \cdot  {\prod\limits_{v\in \{ {\rm p},{\rm d}\} } \left\{ {{\rm d}{\boldsymbol{\rm Q}}_{\boldsymbol{\rm X}_{v}} \cdot {P_{{{\boldsymbol{\rm Q}}_{\boldsymbol{\rm X}_{v}}}}}({{\boldsymbol{\rm Q}}_{\boldsymbol{\rm X}_{v}}})} \right\}}. 
\end{aligned}
\end{equation}
According to the large deviation theorem \cite{ellis2006entropy}, \({{\rm E}_{\boldsymbol{\rm Y}}}{\{ P}_{\boldsymbol{\rm Y}}^\gamma {\rm{(}}{\boldsymbol{\rm Y}}{\rm{)\} }}\) can be further simplified to
\begin{equation}\label{free_energy_appa}
\begin{aligned}
& {{\rm E}_{\boldsymbol{\rm Y}}}{\{ P}_{\boldsymbol{\rm Y}}^\gamma {\rm{(}}{\boldsymbol{\rm Y}}{\rm{)\} }} = \\
& \mathop {\sup }\limits_{\boldsymbol{\rm Q}_{\boldsymbol{\rm X}_{\rm p}},\boldsymbol{\rm Q}_{\boldsymbol{\rm X}_{\rm d}},\boldsymbol{\rm Q}_{\boldsymbol{\rm H}}} \bigg[ {G^{\gamma}(\boldsymbol{\rm Q}_{\boldsymbol{\rm X}_{\rm p}},\boldsymbol{\rm Q}_{\boldsymbol{\rm X}_{\rm d}},\boldsymbol{\rm Q}_{\boldsymbol{\rm H}})}\\
&{\rm{ - }} \sum\limits_{v\in \{ {\rm p},{\rm d} \}} \mathop {\sup }\limits_{\tilde{\boldsymbol{\rm Q}}_{\boldsymbol{\rm X}_{v}}  } \left[ {{\rm tr}(\tilde{\boldsymbol{\rm Q}}_{\boldsymbol{\rm X}_{v}}  {\boldsymbol{\rm Q}}_{\boldsymbol{\rm X}_{v}}   ) - \log M^{\gamma}(\tilde{\boldsymbol{\rm Q}}_{\boldsymbol{\rm X}_{v}} )} \right] \\
&  - \mathop {\sup }\limits_{\tilde{\boldsymbol{\rm Q}}_{\boldsymbol{\rm H}}  } \left[ {{\rm tr}(\tilde{\boldsymbol{\rm Q}}_{\boldsymbol{\rm H}} {\boldsymbol{\rm Q}}_{\boldsymbol{\rm H}}  ) - \log M^{\gamma}(\tilde{\boldsymbol{\rm Q}}_{\boldsymbol{\rm H}} )} \right]  \bigg],
\end{aligned}
\end{equation}
where
\begin{equation}
\begin{aligned}
&{G^{\gamma}(\boldsymbol{\rm Q}_{\boldsymbol{\rm X}_{\rm p}},\boldsymbol{\rm Q}_{\boldsymbol{\rm X}_{\rm d}},\boldsymbol{\rm Q}_{\boldsymbol{\rm H}})} =\\
&\frac{1}{N^2} \log \Bigg\{   {\rm d}{\boldsymbol{\rm Y}} \cdot \prod\limits_{a = 0}^\gamma  \bigg\{ { {P_{\boldsymbol{\rm Y}|\boldsymbol{\rm Z}}}(\boldsymbol{\rm Y}|{\boldsymbol{\rm Z}_{\rm p}^{(a)}},{\boldsymbol{\rm Z}_{\rm d}^{(a)}} )} \\
&\cdot \prod\limits_{v\in \{ {\rm p},{\rm d}\} } \left\{ {\rm d}{{\boldsymbol{\rm Z}}_{v}^{(a)}} \cdot P_{\boldsymbol{\rm Z}_{v}}(\{ {\boldsymbol{\rm Z}_{v}^{(\tilde{a})}}\} _{\tilde{a} = 0}^\gamma|\boldsymbol{\rm Q}_{\boldsymbol{\rm H}},\boldsymbol{\rm Q}_{\boldsymbol{\rm X}_{v}}) \right\} \bigg\}  \Bigg\},
\end{aligned}
\end{equation}
\begin{equation}
\begin{aligned}
& M^{\gamma}(\tilde{\boldsymbol{\rm Q}}_{\boldsymbol{\rm X}_{v}} ) =  {{\rm E}_{\{ {\boldsymbol{\rm X}_{v}^{(a)}}\} _{a = 0}^\gamma }}\left\{ \prod\limits_{k=1}^{K_{v}} {{e^{{\rm tr}({{\tilde{\boldsymbol{\rm Q}}}_{\boldsymbol{\rm X}_{v}}} \boldsymbol{\rm X}_{v,k}^H\boldsymbol{\rm X}_{v,k} )}}} \right\},\\
&\quad\quad\quad\quad\quad\quad\quad\quad\quad\quad\quad\quad\quad\quad\quad\quad \forall v\in \{{\rm p},{\rm d}\}
\end{aligned}
\end{equation}
\begin{equation}
\begin{aligned}
M^{\gamma}(\tilde{\boldsymbol{\rm Q}}_{\boldsymbol{\rm H}} ) =  {{\rm E}_{\{ {\boldsymbol{\rm H}^{(a)}}\} _{a = 0}^\gamma }}\left\{ \prod\limits_{m=1}^{M} {{e^{{\rm tr}({{\tilde{\boldsymbol{\rm Q}}}_{\boldsymbol{\rm H}}} \boldsymbol{\rm H}_m^H \boldsymbol{\rm H}_m )}}} \right\}  ,
\end{aligned}
\end{equation}
\( \boldsymbol{\rm H}_{m} = [{\boldsymbol{\rm h}_{m}^{(0)}}^T, {\boldsymbol{\rm h}_{m}^{(1)}}^T, \ ...\ ,{\boldsymbol{\rm h}_{m}^{(\gamma)}}^{T}] \), \( \boldsymbol{\rm X}_{v,k} = [{\boldsymbol{\rm x}_{v,k}^{(0)}}, {\boldsymbol{\rm x}_{v,k}^{(1)}}, \ ...\ ,{\boldsymbol{\rm x}_{v,k}^{(\gamma)}}] \). \( {\boldsymbol{\rm h}_{m}^{(a)}} \) is the \(m\)th row vector of \(\boldsymbol{\rm H}^{(a)}\), and \( {\boldsymbol{\rm x}_{k}^{(b)}} \) is the \(k\)th column vector of \(\boldsymbol{\rm X}_{v}^{(b)},\quad \forall v \in \{{\rm p},{\rm d}\}\). \( M^{\gamma}(\tilde{\boldsymbol{\rm Q}}_{\boldsymbol{\rm H}} ) \) and \( M^{\gamma}(\tilde{\boldsymbol{\rm Q}}_{\boldsymbol{\rm X}_{v}} )  \) are the moment generating functions of \( \{ \boldsymbol{\rm H}^{(a)} \}_{a=1}^{\gamma} \) and \( \{\boldsymbol{\rm X}_{v}^{(a)} \}_{a=1}^{\gamma} \) respectively. The auxiliary
matrices \( \tilde{\boldsymbol{\rm Q}}_{\boldsymbol{\rm H}} \in {{\mathbb{C}}^{(\gamma+1) \times (\gamma+1)}} \) and \( \tilde{\boldsymbol{\rm Q}}_{\boldsymbol{\rm X}_{v}}  \in {{\mathbb{C}}^{(\gamma+1) \times (\gamma+1)}} \) are introduced for the
saddle point method which is used for calculating the free energy. 

Considering the difficulty in directly solving for the saddle point of the free energy \eqref{free_energy_appa}, we assume that the saddle point follows the replica symmetry form \cite{Kabashima2016}, which can be given by
\begin{equation}
\boldsymbol{\rm Q}_{\boldsymbol{\rm H}} = (c_{\boldsymbol{\rm H}} - q_{\boldsymbol{\rm H}})\boldsymbol{\rm I}_{\gamma+1} + q_{\boldsymbol{\rm H}}\boldsymbol{\rm 1}_{\gamma+1},
\end{equation}
\begin{equation}
\tilde{\boldsymbol{\rm Q}}_{\boldsymbol{\rm H}} = (\tilde{c}_{\boldsymbol{\rm H}} - \tilde{q}_{\boldsymbol{\rm H}})\boldsymbol{\rm I}_{\gamma+1} + {\tilde{q}}_{\boldsymbol{\rm H}}\boldsymbol{\rm 1}_{\gamma+1},
\end{equation}
\begin{equation}
\boldsymbol{\rm Q}_{\boldsymbol{\rm X}_{v}} = (c_{\boldsymbol{\rm X}_{v}} - q_{\boldsymbol{\rm X}_{v}})\boldsymbol{\rm I}_{\gamma+1} + q_{\boldsymbol{\rm X}_{v}}\boldsymbol{\rm 1}_{\gamma+1},\quad \forall v \in \{{\rm p},{\rm d}\}
\end{equation}
\begin{equation}
\tilde{\boldsymbol{\rm Q}}_{\boldsymbol{\rm X}_{v}} = (\tilde{c}_{\boldsymbol{\rm X}_{v}} - \tilde{q}_{\boldsymbol{\rm X}_{v}})\boldsymbol{\rm I}_{\gamma+1} + \tilde{q}_{\boldsymbol{\rm X}_{v}}\boldsymbol{\rm 1}_{\gamma+1}.\quad \forall v \in \{{\rm p},{\rm d}\}
\end{equation}

From \eqref{app_qz}, \(z_{m,k}^{(a)}\) can be modeled as
\begin{equation}
z_{m,k}^{(a)} =  \sqrt{c_{\boldsymbol{\rm H}} c_{\boldsymbol{\rm X}_{v}} 
 - q_{\boldsymbol{\rm H}} q_{\boldsymbol{\rm X}_{v}}  } b_{m,k}^{(a)} + \sqrt{ q_{\boldsymbol{\rm H}} q_{\boldsymbol{\rm X}_{v}} } \tilde{b}_{m,k},\quad \forall v \in \{{\rm p},{\rm d}\}
\end{equation} 
where \(\tilde{b}_{m,k}\) and \( \{ b_{m,k}^{(a)} \}_{a=0}^{\gamma} \) are independent standard Gaussian variables. Then, the free energy \eqref{free_energy_appa} can be decomposed into several terms for calculation,
\begin{equation}\label{app_de_1}
\begin{aligned}
&{G^{\gamma}(\boldsymbol{\rm Q}_{\boldsymbol{\rm X}_{\rm p}},\boldsymbol{\rm Q}_{\boldsymbol{\rm X}_{\rm d}},\boldsymbol{\rm Q}_{\boldsymbol{\rm H}})} = \\
&\alpha \cdot \sum\limits_{v\in \{{\rm p},{\rm d}\} } \Bigg\{ \beta_{v} \cdot \log \bigg( \int {{\rm d}y \cdot {\mathcal{D}}{\tilde{b}} \cdot \prod\limits_{a = 0}^\gamma  \Big\{ {\mathcal{D}{b^{(a)}} \cdot}}\\
& {P_{y|z}}(y|z = \sqrt{c_{\boldsymbol{\rm H}} c_{\boldsymbol{\rm X}_{v}} 
 - q_{\boldsymbol{\rm H}} q_{\boldsymbol{\rm X}_{v}}  } b^{(a)} + \sqrt{ q_{\boldsymbol{\rm H}} q_{\boldsymbol{\rm X}_{v}} } \tilde{b} ) \Big\}  \bigg) \Bigg\}, 
\end{aligned}
\end{equation}
where \(P_{y|z}({y|z})= \mathcal{CN}(y;z,\sigma_{\rm n}^2)\). \(\mathcal{D}\xi\) is a Gaussian integral \({\rm d}\xi \cdot e^{-\frac{\xi^2}{2}} / {\sqrt{2\pi}} \). After utilizing a formula \({e^{A\sum\nolimits_{{\rm{a}} < b} {{w_a}{w_b}} }} = {e^{ - \frac{A}{2}\sum\nolimits_{\rm{a}} {{{({w_a})}^2}} }}\int {\mathcal{D}\xi  \cdot {e^{\sqrt A \xi \sum\nolimits_{\rm{a}} {{w_a}} }}} \), the other terms in \eqref{free_energy_appa} can be calculated as
\begin{equation}\label{app_de_2}
\begin{aligned}
&{\rm tr}(\tilde{\boldsymbol{\rm Q}}_{\boldsymbol{\rm X}_{v}}  {\boldsymbol{\rm Q}}_{\boldsymbol{\rm X}_{v}}   ) =  (\gamma+1)\tilde{c}_{\boldsymbol{\rm X}_{v}}c_{\boldsymbol{\rm X}_{v}}  + \gamma(\gamma+1) \tilde{q}_{\boldsymbol{\rm X}_{v}} q_{\boldsymbol{\rm X}_{v}},\\
&\quad\quad\quad\quad\quad\quad\quad\quad\quad\quad\quad\quad\quad\quad\quad\quad\quad \forall v\in \{{\rm p},{\rm d}\}
\end{aligned}
\end{equation}
\begin{equation}\label{app_de_3}
\begin{aligned}
{\rm tr}(\tilde{\boldsymbol{\rm Q}}_{\boldsymbol{\rm H}}  {\boldsymbol{\rm Q}}_{\boldsymbol{\rm H}}   ) =  (\gamma+1)\tilde{c}_{\boldsymbol{\rm H}}c_{\boldsymbol{\rm H}}  + \gamma(\gamma+1) \tilde{q}_{\boldsymbol{\rm H}} q_{\boldsymbol{\rm H}} ,
\end{aligned}
\end{equation}
\begin{equation}\label{app_de_5}
\begin{aligned}
&M^{\gamma}(\tilde{\boldsymbol{\rm Q}}_{\boldsymbol{\rm X}_{v}} ) = \frac{1}{N^2} \log \bigg( {\rm E}_{\{{\boldsymbol{\rm X}_{v}^{(a)}}\}_{a=0}^{\gamma}} \bigg\{   \prod\limits_{k=1}^{K_{v}}\prod\limits_{n=1}^{N}  \mathcal{D}{\xi_{n,k}} \\
&\cdot e^{2{\tilde{q}_{\boldsymbol{\rm X}_{v}}}{\rm Re}\left( \sum\limits_{a=0}^{\gamma} \xi_{n,k} \cdot x_{v,n,k}^{(a)} \right)} \cdot  e^{(\tilde{c}_{\boldsymbol{\rm X}_{v}} - \tilde{q}_{\boldsymbol{\rm X}_{v}}  )  \sum\limits_{a=0}^{\gamma} \left| x_{v,n,k}^{(a)} \right|^2 } \bigg\} \bigg) \\
& = \beta_{v} \cdot \log \bigg(  {\rm E}_{x_{v}} \left\{ \int  {\rm d}\xi \cdot  e^{-\left| \xi -\sqrt{\tilde{q}_{\boldsymbol{\rm X}_{v}}  } x_{v} \right|^2 }  \right\}^{\gamma+1}   \bigg) ,\\
&\quad\quad\quad\quad\quad\quad\quad\quad\quad\quad\quad\quad\quad\quad\quad\quad\quad \forall v\in \{{\rm p},{\rm d}\}
\end{aligned}
\end{equation}
\begin{equation}\label{app_de_4}
\begin{aligned}
&M^{\gamma}(\tilde{\boldsymbol{\rm Q}}_{\boldsymbol{\rm H}} ) = \frac{1}{N^2} \log \bigg( {\rm E}_{\{{\boldsymbol{\rm H}^{(a)}}\}_{a=0}^{\gamma}} \bigg\{   \prod\limits_{m=1}^{M}\prod\limits_{n=1}^{N}  \mathcal{D}{\xi_{m,n}} \\
&\cdot e^{2{\tilde{q}_{\boldsymbol{\rm H}}}{\rm Re}\left( \sum\limits_{a=0}^{\gamma} \xi_{m,n} h_{m,n}^{(a)} \right)} \cdot  e^{(\tilde{c}_{\boldsymbol{\rm H}} - \tilde{q}_{\boldsymbol{\rm H}}  )  \sum\limits_{a=0}^{\gamma} \left| h_{m,n}^{(a)} \right|^2 } \bigg\} \bigg) \\
& = \alpha \log \bigg(  {\rm E}_{h} \left\{ \int  {\rm d}\xi \cdot  e^{-\left| \xi -\sqrt{\tilde{q}_{\boldsymbol{\rm H}}  } h \right|^2 }  \right\}^{\gamma+1}   \bigg) .
\end{aligned}
\end{equation}
Substituting \eqref{app_de_1} to \eqref{app_de_5} into equation \eqref{free_energy_appa} and setting the partial derivatives of the free energy \eqref{free_energy_appa} to zero, when \(\gamma \to 0\), \(\tilde{c}_{\boldsymbol{\rm H}}=0\), \(\tilde{c}_{\boldsymbol{\rm X}_{v}}=0,\  \forall v \in \{{\rm p},{\rm d}\}\) and the fixed point equations \eqref{Q_tilde_Q_1} to \eqref{Q_tilde_Q_6} can be obtained.

\section{The Proof of Lemma 2}\label{App_B}
According to Carleman theorem, in order to prove the validity of Lemma \ref{Lemma 2}, it is first necessary to demonstrate that the joint moment of \(\{{h}_{m,n}^{(0)},{\hat{h}}_{m,n}^{(0)},{x_{{\rm d},\tilde{n},k}^{(0)}},{\hat{x}_{{\rm d},\tilde{n},k}^{(0)}}\}\) converges to the joint moment of \(\{{h},{\hat{h}},{x_{{\rm d}}},{\hat{x}_{{\rm d}}}\}\) as \(N \to +\infty\). Note that \(\{{h},{\hat{h}},{x_{{\rm d}}},{\hat{x}_{{\rm d}}}\}\) is the real values and posterior mean estimation results from the AGWN channels \eqref{X_d_AGMN_pro_2} and \eqref{H_AGMN_pro_2}. Then, we first calculates the joint moment of \( \{  {\rm Re}(h_{m,n}^{(0)}),{\rm Im}(h_{m,n}^{(0)}), {\rm Re}(h_{m,n}^{(a_{\rm h})}),{\rm Im}(h_{m,n}^{(a_{\rm h})}),\) \({\rm Re}(x_{{\rm d},\tilde{n},k}^{(0)}), {\rm Im}(x_{{\rm d},\tilde{n},k}^{(0)}), {\rm Re}(x_{{\rm d},\tilde{n},k}^{(a_{\rm x})}),{\rm Im}(x_{{\rm d},\tilde{n},k}^{(a_{\rm x})})  \} \) by [\cite{Guo2005}, Lemma 1],
\begin{equation}\label{joint_moment_sum}
\begin{aligned}
&\sum\limits_{m = 1,n=1}^{M_1,N_1} \sum\limits_{\tilde{n} = 1,k=1}^{\tilde{N}_1,K_1} {\rm E}\{{\rm Re}(h_{m,n}^{(0)})^{i_{\rm R_h}}\cdot{\rm Im}(h_{m,n}^{(0)})^{i_{\rm I_h}}\cdot {\rm Re}(h_{m,n}^{(a_{\rm h})})^{j_{\rm R_h}} \\
&\cdot{\rm Im}(h_{m,n}^{(a_{\rm h})})^{j_{\rm I_h}}\cdot{\rm Re}(x_{{\rm d},\tilde{n},k}^{(0)})^{i_{\rm R_{\rm x}}}\cdot {\rm Im}(x_{{\rm d},\tilde{n},k}^{(0)})^{i_{\rm I_{\rm x}}}\cdot {\rm Re}(x_{{\rm d},\tilde{n},k}^{(a_{\rm x})})^{j_{\rm R_{\rm x}}}\\
&\cdot{\rm Im}(x_{{\rm d},\tilde{n},k}^{(a_{\rm x})})^{j_{\rm I_{\rm x}}}\}  = \\
&\mathop {{\rm{lim}}}\limits_{\gamma \to 0 } \frac{\partial^2}{\partial \psi_{\rm h} \partial \psi_{\rm x} } \log {{\rm E}_{\boldsymbol{\rm Y}}}\{  e^{\psi_{\rm h} \cdot g_{\rm h} \cdot \psi_{\rm x} \cdot g_{\rm x} } \cdot {P}_{\boldsymbol{\rm Y}}^\gamma {\rm{(}}{\boldsymbol{\rm Y}}{\rm{)\}\Big|_{\psi_{\rm x}=0,\psi_{\rm h} =0 } }},
\end{aligned}
\end{equation}
where  
\begin{equation}
\begin{aligned}
&g_{\rm h} = \sum\limits_{m = 1,n=1}^{M_1,N_1} {\rm Re}(h_{m,n}^{(0)})^{i_{\rm R_h}}\cdot{\rm Im}(h_{m,n}^{(0)})^{i_{\rm I_h}}\cdot {\rm Re}(h_{m,n}^{(a_{\rm h})})^{j_{\rm R_h}}\\
&\cdot{\rm Im}(h_{m,n}^{(a_{\rm h})})^{j_{\rm I_h}},
\end{aligned}
\end{equation}
\begin{equation}
\begin{aligned}
&g_{\rm x} =\sum\limits_{\tilde{n} = 1,k=1}^{\tilde{N}_1,K_1} {\rm Re}(x_{{\rm d},\tilde{n},k}^{(0)})^{i_{\rm R_{\rm x}}}\cdot {\rm Im}(x_{{\rm d},\tilde{n},k}^{(0)})^{i_{\rm I_{\rm x}}}\cdot {\rm Re}(x_{{\rm d},\tilde{n},k}^{(a_{\rm x})})^{j_{\rm R_{\rm x}}}\\
&\cdot{\rm Im}(x_{{\rm d},\tilde{n},k}^{(a_{\rm x})})^{j_{\rm I_{\rm x}}}.
\end{aligned}
\end{equation}
\(M_1,N_1,\tilde{N}_1, K_1\) are the parameters which are less than \(M,N,N,K_{\rm d}\) and grows linearly with \(M,N,N,K_{\rm d}\) respectively. Similar to \eqref{free_energy_appa}, under the large system condition, \eqref{joint_moment_sum} can be further calculated as
\begin{equation}\label{free_energy_appa}
\begin{aligned}
& {{\rm E}_{\boldsymbol{\rm Y}}}{\{ P}_{\boldsymbol{\rm Y}}^\gamma {\rm{(}}{\boldsymbol{\rm Y}}{\rm{)\} }} = \\
    & \mathop {\sup }\limits_{\boldsymbol{\rm Q}_{\boldsymbol{\rm X}_{\rm p}},\boldsymbol{\rm Q}_{\boldsymbol{\rm X}_{\rm d}},\boldsymbol{\rm Q}_{\boldsymbol{\rm H}}} \Bigg[ {G^{\gamma}(\boldsymbol{\rm Q}_{\boldsymbol{\rm X}_{\rm p}},\boldsymbol{\rm Q}_{\boldsymbol{\rm X}_{\rm d}},\boldsymbol{\rm Q}_{\boldsymbol{\rm H}})} \\
&{\rm{ - }}  \sum\limits_{v \in \{{\rm p,d}\}} \mathop {\sup }\limits_{\tilde{\boldsymbol{\rm Q}}_{\boldsymbol{\rm X}_{v}}  } \left[ {{\rm tr}(\tilde{\boldsymbol{\rm Q}}_{\boldsymbol{\rm X}_{v}}  {\boldsymbol{\rm Q}}_{\boldsymbol{\rm X}_{v}}   ) - \log \tilde{M}^{\gamma}(\tilde{\boldsymbol{\rm Q}}_{\boldsymbol{\rm X}_{v}}  )} \right] \\
&  - \mathop {\sup }\limits_{\tilde{\boldsymbol{\rm Q}}_{\boldsymbol{\rm H}}  } \left[ {{\rm tr}(\tilde{\boldsymbol{\rm Q}}_{\boldsymbol{\rm H}} {\boldsymbol{\rm Q}}_{\boldsymbol{\rm H}}  ) - \log \tilde{M}^{\gamma}(\tilde{\boldsymbol{\rm Q}}_{\boldsymbol{\rm H}} )} \right] \\
& -  \left(  \log \tilde{M}^{\gamma}(\tilde{\boldsymbol{\rm Q}}_{\boldsymbol{\rm H}}, \tilde{\boldsymbol{\rm Q}}_{\boldsymbol{\rm X}_{\rm d}}, \psi_{\rm x}, \psi_{\rm h})   \right)   \Bigg],
\end{aligned}
\end{equation}
where 
\begin{equation}\label{energy_appb}
\begin{aligned}
&\log \tilde{M}^{\gamma}(\tilde{\boldsymbol{\rm Q}}_{\boldsymbol{\rm H}}, \tilde{\boldsymbol{\rm Q}}_{\boldsymbol{\rm X}_{\rm d}}, \psi_{\rm x}, \psi_{\rm h}) \\
& = \Big( \log {\rm E}\{ e^{\psi_{\rm x}\cdot \tilde{g}_{\rm x} \cdot \psi_{\rm h}\cdot \tilde{g}_{\rm h}  }\cdot e^{\tilde{\boldsymbol{\rm x}}_{\rm d}^{H} \tilde{\boldsymbol{\rm Q}}_{\boldsymbol{\rm X}_{\rm d}} \tilde{\boldsymbol{\rm x}}_{\rm d}  } \cdot e^{\tilde{\boldsymbol{\rm h}}^{H} \tilde{\boldsymbol{\rm Q}}_{\boldsymbol{\rm H}} \tilde{\boldsymbol{\rm h}} } \} -  \log M^{\gamma}(\tilde{\boldsymbol{\rm Q}}_{\boldsymbol{\rm H}} ) \\
& - \log M^{\gamma}(\tilde{\boldsymbol{\rm Q}}_{\boldsymbol{\rm X}_{\rm d}} ) \Big),
\end{aligned}
\end{equation}
\begin{equation}
\begin{aligned}
&\tilde{g}_{\rm h} = {\rm Re}(\tilde{h}^{(0)})^{i_{\rm R_h}}\cdot{\rm Im}(\tilde{h}^{(0)})^{i_{\rm I_h}}\cdot {\rm Re}(\tilde{h}^{(a_{\rm h})})^{j_{\rm R_h}}\cdot{\rm Im}(\tilde{h}^{(a_{\rm h})})^{j_{\rm I_h}},
\end{aligned}
\end{equation}
\begin{equation}
\begin{aligned}
&\tilde{g}_{\rm x} = {\rm Re}(\tilde{x}_{{\rm d}}^{(0)})^{i_{\rm R_{\rm x}}}\cdot {\rm Im}(\tilde{x}_{{\rm d}}^{(0)})^{i_{\rm I_{\rm x}}}\cdot {\rm Re}(\tilde{x}_{{\rm d}}^{(a_{\rm x})})^{j_{\rm R_{\rm x}}}\cdot{\rm Im}(\tilde{x}_{{\rm d}}^{(a_{\rm x})})^{j_{\rm I_{\rm x}}}.
\end{aligned}
\end{equation}
\(\tilde{\boldsymbol{\rm x}}_{\rm d} = [\tilde{x}^{(0)}_{\rm d},\tilde{x}^{(1)}_{\rm d},\tilde{x}^{(2)}_{\rm d}, \ ...\ , \tilde{x}^{(\gamma)}_{\rm d}]\) and \(\tilde{\boldsymbol{\rm h}} = [\tilde{h}^{(0)},\tilde{h}^{(1)},\tilde{h}^{(2)}, \ ...\ , \tilde{h}^{(\gamma)}]\) are independent random variables whose distributions follow \(\tilde{x}^{(0)}_{\rm d},\tilde{x}^{(1)}_{\rm d},\tilde{x}^{(2)}_{\rm d}, \ ...\ , \tilde{x}^{(\gamma)}_{\rm d} \to P_{x}\) and \(\tilde{h}^{(0)},\tilde{h}^{(1)},\tilde{h}^{(2)}, \ ...\ , \tilde{h}^{(\gamma)} \to P_{h}\). Taking the derivative in \eqref{energy_appb} with respect to \(\psi_{\rm h}\) and \(\psi_{\rm x}\) at \(\psi_{\rm x} = 0,\psi_{\rm h} = 0 \) leaves only one term
\begin{equation}\label{proof_1}
\begin{aligned}
&\frac{\partial^{2}}{\partial \psi_{\rm x}\cdot\partial \psi_{\rm h}} \log \tilde{M}^{\gamma}(\tilde{\boldsymbol{\rm Q}}_{\boldsymbol{\rm H}}, \tilde{\boldsymbol{\rm Q}}_{\boldsymbol{\rm X}_{\rm d}}, \psi_{\rm x}, \psi_{\rm h})\Big|_{\psi_{\rm x} = 0, \psi_{\rm h} = 0} 
 \\
& = \frac{{\rm E}\{ \tilde{g}_{\rm x} \cdot \tilde{g}_{\rm h}  \cdot e^{\tilde{\boldsymbol{\rm x}}^{H} \tilde{\boldsymbol{\rm Q}}_{\boldsymbol{\rm X}_{\rm d}} \tilde{\boldsymbol{\rm x}}  } \cdot e^{\tilde{\boldsymbol{\rm h}}^{H} \tilde{\boldsymbol{\rm Q}}_{\boldsymbol{\rm H}} \tilde{\boldsymbol{\rm h}} } \}}{{\rm E}\{ e^{\tilde{\boldsymbol{\rm x}}^{H} \tilde{\boldsymbol{\rm Q}}_{\boldsymbol{\rm X}_{\rm d}} \tilde{\boldsymbol{\rm x}}  } \cdot e^{\tilde{\boldsymbol{\rm h}}^{H} \tilde{\boldsymbol{\rm Q}}_{\boldsymbol{\rm H}} \tilde{\boldsymbol{\rm h}} } \}}\\
& = \int\limits P_{y_{{x}_{\rm d}}}(y_{{x}_{\rm d}}) \cdot {\rm E}\{ {\rm Re}(\tilde{x}_{{\rm d}}^{(0)})^{i_{\rm R_{\rm x}}}\cdot {\rm Im}(\tilde{x}_{{\rm d}}^{(0)})^{i_{\rm I_{\rm x}}} |y_{x_{\rm d}} \}  \\
&\cdot{\rm E}\{ {\rm Re}(\tilde{x}_{{\rm d}}^{(a_{\rm x})})^{j_{\rm R_{\rm x}}}\cdot{\rm Im}(\tilde{x}_{{\rm d}}^{(a_{\rm x})})^{j_{\rm I_{\rm x}}} |{y_{x_{\rm d}}}\}  {\rm d}{y_{x_{\rm d}}}\\
& \cdot \int\limits P_{y_{h}}(y_{h}) \cdot {\rm E}\{ {\rm Re}(\tilde{h}^{(0)})^{i_{\rm R_h}}\cdot{\rm Im}(\tilde{h}^{(0)})^{i_{\rm I_h}} |y_{h} \}  \\
&\cdot{\rm E}\{ {\rm Re}(\tilde{h}^{(a_{\rm h})})^{j_{\rm R_h}}\cdot{\rm Im}(\tilde{h}^{(a_{\rm h})})^{j_{\rm I_h}} |y_{h}\}   {\rm d}{y_{h}}, \quad ({\rm as}\  \gamma \to 0 )
\end{aligned}
\end{equation}
where \(y_{x_{\rm d}}\) and \(y_{h}\) are the output of the scalar AGWN channels \eqref{X_d_AGMN_pro_2} and \eqref{H_AGMN_pro_2}. From \eqref{proof_1}, it can be proven that under the large system conditions, the joint moments of \( \{  {\rm Re}(h_{m,n}^{(0)}),{\rm Im}(h_{m,n}^{(0)}), {\rm Re}(h_{m,n}^{(a_{\rm h})}),{\rm Im}(h_{m,n}^{(a_{\rm h})}), {\rm Re}(x_{{\rm d},n,k}^{(0)}),\) \( {\rm Im}(x_{{\rm d},n,k}^{(0)}), {\rm Re}(x_{{\rm d},m,n}^{(a_{\rm x})}),{\rm Im}(x_{{\rm d},m,n}^{(a_{\rm x})})  \} \) are independent of \(m,n,k,a_{\rm x}\) and \(a_{\rm h}\), which can be given by
\begin{equation}
\begin{aligned}
&{\rm E}\{{\rm Re}(h_{m,n}^{(0)})^{i_{\rm R_h}}\cdot{\rm Im}(h_{m,n}^{(0)})^{i_{\rm I_h}}\cdot {\rm Re}(h_{m,n}^{(a_{\rm h})})^{j_{\rm R_h}}\cdot{\rm Im}(h_{m,n}^{(a_{\rm h})})^{j_{\rm I_h}}\\
&\cdot{\rm Re}(x_{{\rm d},n,k}^{(0)})^{i_{\rm R_{\rm x}}}\cdot {\rm Im}(x_{{\rm d},n,k}^{(0)})^{i_{\rm I_{\rm x}}}\cdot {\rm Re}(x_{{\rm d},n,k}^{(a_{\rm x})})^{j_{\rm R_{\rm x}}} \\  
& \cdot {\rm Im}(x_{{\rm d},n,k}^{(a_{\rm x})})^{j_{\rm I_{\rm x}}} \} \to\\
&{\rm E}\{{\rm Re}(h)^{i_{\rm R_h}}\cdot{\rm Im}(h)^{i_{\rm I_h}}\cdot {\rm Re}(\tilde{h})^{j_{\rm R_h}}  \cdot{\rm Im}(\tilde{h})^{j_{\rm I_h}}\\
& \cdot{\rm Re}(x_{{\rm d}})^{i_{\rm R_{\rm x}}}\cdot {\rm Im}(x_{{\rm d}})^{i_{\rm I_{\rm x}}}\cdot {\rm Re}(\tilde{x}_{{\rm d}})^{j_{\rm R_{\rm x}}} \cdot {\rm Im}(\tilde{x}_{{\rm d}})^{j_{\rm I_{\rm x}}} \},
\end{aligned}
\end{equation}
where \( \tilde{x}_{\rm d}\) and \(\tilde{h}\) are the output of retrochannels \cite{Guo2005} of the scalar AGWN channels \eqref{X_d_AGMN_pro_2} and \eqref{H_AGMN_pro_2}. The proof results indicate that, for each parameters \(m\), \(n\), \(\tilde{n}\), and \(k\), the joint distribution of \(\{ h_{m,n}^{(0)}, h_{m,n}^{(a_{\rm h})}, x_{{\rm d},\tilde{n},k}^{(0)},x_{{\rm d},\tilde{n},k}^{(a_{\rm x})}  \}\) converges
to the joint distribution of \( \{ h, \tilde{h}, x_{\rm d}, \tilde{x}_{\rm d} \} \). Furthermore, because replicas \( \tilde{x}_{\rm d}\) and \(\tilde{h}\) are random variables following the posterior PDF \(P_{x_{\rm d}|y_{x_{\rm d}}}\) and \(P_{h|y_{h}}\), it can be easily proven that, 
\begin{equation}
\begin{aligned}
&{\rm E}\{{\rm Re}(h_{m,n}^{(0)})^{i_{\rm R_h}}\cdot{\rm Im}(h_{m,n}^{(0)})^{i_{\rm I_h}}\cdot {\rm Re}(\hat{h}_{m,n})^{j_{\rm R_h}}\cdot{\rm Im}(\hat{h}_{m,n})^{j_{\rm I_h}}\\
&\cdot{\rm Re}(x_{{\rm d},\tilde{n},k}^{(0)})^{i_{\rm R_{\rm x}}}\cdot {\rm Im}(x_{{\rm d},\tilde{n},k}^{(0)})^{i_{\rm I_{\rm x}}}\cdot {\rm Re}(\hat{x}_{{\rm d},\tilde{n},k})^{j_{\rm R_{\rm x}}} \\  
& \cdot {\rm Im}(\hat{x}_{{\rm d},\tilde{n},k})^{j_{\rm I_{\rm x}}} \} \to\\
&{\rm E}\{{\rm Re}(h)^{i_{\rm R_h}}\cdot{\rm Im}(h)^{i_{\rm I_h}}\cdot {\rm Re}(\hat{h})^{j_{\rm R_h}}  \cdot{\rm Im}(\hat{h})^{j_{\rm I_h}}\\
& \cdot{\rm Re}(x_{{\rm d}})^{i_{\rm R_{\rm x}}}\cdot {\rm Im}(x_{{\rm d}})^{i_{\rm I_{\rm x}}}\cdot {\rm Re}(\hat{x}_{{\rm d}})^{j_{\rm R_{\rm x}}} \cdot {\rm Im}(\hat{x}_{{\rm d}})^{j_{\rm I_{\rm x}}} \}.
\end{aligned}
\end{equation}
This implies that, under the condition of a large system, the original estimation problem can be equivalent to the independent scalar AGWN channel estimation problems \eqref{X_d_AGMN_pro_2} and \eqref{H_AGMN_pro_2}.

\section{The Proof of Proposition 1}\label{App_C}
First, we calculate the posterior probability distributions of \(P_{{x}_{\rm d}|y_{{x}_{\rm d}}}({x}_{\rm d}|y_{{x}_{\rm d}})\) \eqref{x_d_PP_pro2} and \(P_{{h}|y_{{h}}}({h}|y_{{h}})\) \eqref{h_PP_pro2}. In scalar Gaussian channels \eqref{X_d_AGMN_pro_2} and \eqref{H_AGMN_pro_2}, \(P_{y_{{x}_{\rm d}}|{x}_{\rm d}}(y_{{x}_{\rm d}}|{x}_{\rm d})\) and \(P_{y_{h}|h}(y_{h}|h) \) can be given by
\begin{align}
& P_{y_{{x}_{\rm d}}|{x}_{\rm d}}(y_{{x}_{\rm d}}|{x}_{\rm d})  =  \mathcal{CN}( y_{x_{\rm d}};\sqrt{\tilde{q}_{\boldsymbol{\rm X}_{\rm d}}} x_{\rm d}, 1),\\
& P_{y_{h}|h}(y_{h}|h)  =  \mathcal{CN}( y_{h};\sqrt{\tilde{q}_{\boldsymbol{\rm H}}} h, 1).
\end{align}
Note that \(P_{x_{\rm d}}(x_{\rm d})\) follows a Gaussian distribution (i.e. Gaussian codebook) \eqref{X_G_Distribution}, and \(P_{h}(h)\) follows a BG distribution \eqref{BG_Distribution}. For simplicity, we assume that each user has the same parameters \(\lambda_{n}=\lambda\) and \(\gamma_{n}=\sigma_{{\rm h}}^2,\quad\forall n\in \{1,2,\ ...\ ,N\} \) in the BG distribution. The generalized version can be obtained by decomposing matrix \(\boldsymbol{\rm H}\) into column vectors \(\{\boldsymbol{\rm h}_{n}\}_{n=1}^{N}\), but this becomes too cumbersome. Through some derivations\footnote{
The multiplication of two Gaussian distributions can result in a new Gaussian distribution, which is given by
\begin{equation}
\frac{1}{C} \mathcal{CN}(a;\mu_1,\sigma_1^2) \cdot \mathcal{CN}(a;\mu_2,\sigma_2^2) = \mathcal{CN}(a;\frac{\mu_2\sigma_1^2+\mu_1\sigma_2^2}{\sigma_1^2+\sigma_2^2},\frac{\sigma_1^2\sigma_2^2}{\sigma_1^2+\sigma_2^2}),
\end{equation}
where \(C\) ensures that the PDF integrates to one,
\begin{equation}
C = \frac{1}{\sqrt{2\pi(\sigma_1^2+\sigma_2^2)}}e^{-\frac{(\mu_1-\mu_2)^2}{2(\sigma_1^2+\sigma_2^2)}}.
\end{equation}
Also, the relationship can be proven,
\begin{equation}
\mathcal{CN}(a;\epsilon \cdot b,\sigma^2) = \mathcal{CN}(b; \frac{a}{\epsilon}, \frac{\sigma^2}{\epsilon^2}),
\end{equation}
where \(\epsilon\) is a constant, and \(a\) and \(b\) are variables.
}, \(P_{{x}_{\rm d}|y_{{x}_{\rm d}}}({x}_{\rm d}|y_{{x}_{\rm d}})\) and \(P_{{h}|y_{{h}}}({h}|y_{{h}})\) can be given by
\begin{equation}\label{app_C_post_pdf_x_d}
P_{{x}_{\rm d}|y_{{x}_{\rm d}}}({x}_{\rm d}|y_{{x}_{\rm d}}) = \mathcal{CN}(x_{\rm d};\frac{\sqrt{\tilde{q}_{\boldsymbol{\rm X}_{\rm d}}}\sigma_{\rm x}^2 y_{x_{\rm d}}}{1+\tilde{q}_{\boldsymbol{\rm X}_{\rm d}}\sigma_{\rm x}^{2}},\frac{\sigma_{\rm x}^2}{1+\tilde{q}_{\boldsymbol{\rm X}_{\rm d}}\sigma_{\rm x}^2}),   
\end{equation}
\begin{equation}\label{app_C_post_pdf_h}
\begin{aligned}
&P_{\tilde{h}|y_{{h}}}({h}|y_{{h}}) = \frac{(1 - \lambda){C_1}}{(1 - \lambda){C_1} + \lambda {C_2}}\delta(h)\\
&\quad\quad\quad + \frac{\lambda {C_2}}{(1 - \lambda){C_1} + \lambda {C_2}} \mathcal{CN}(h;\frac{\sqrt{\tilde{q}_{\boldsymbol{\rm H}}}\sigma_{\rm h}^2 y_{h}}{1+\tilde{q}_{\boldsymbol{\rm H}}\sigma_{\rm h}^{2}},\frac{\sigma_{\rm h}^2}{1+\tilde{q}_{\boldsymbol{\rm H}}\sigma_{\rm h}^2}),
\end{aligned}   
\end{equation}
where
\begin{equation}
C_1 = \frac{1}{\sqrt{2\pi(\tilde{q}_{\boldsymbol{\rm H}})^{-1}}}e^{-\frac{\left\|y_{h}\right\|^2}{2}}
\end{equation}
and 
\begin{equation}
C_2 = \frac{1}{\sqrt{2\pi[\sigma_{\rm h}^2+(\tilde{q}_{\boldsymbol{\rm H}})^{-1}]}}e^{-\frac{\left\|y_{h}\right\|^2}{2(\tilde{q}_{\boldsymbol{\rm H}}\sigma_{\rm h}^2+1)}}.
\end{equation}
Then, from \eqref{app_C_post_pdf_x_d} and \eqref{app_C_post_pdf_h}, the posterior mean \(\hat{x}_{\rm d}\) and \(\hat{h}\) can be given by
\begin{align}
& \hat{x}_{\rm d} = \frac{\sqrt{\tilde{q}_{\boldsymbol{\rm X}_{\rm d}}}\sigma_{\rm x}^2 y_{x_{\rm d}}}{1+\tilde{q}_{\boldsymbol{\rm X}_{\rm d}}\sigma_{\rm x}^{2}},\\ 
& \hat{h} = \frac{\lambda {C_2}}{(1 - \lambda){C_1} + \lambda {C_2}} \cdot \frac{\sqrt{\tilde{q}_{\boldsymbol{\rm H}}}\sigma_{\rm h}^2 y_{h}}{1+\tilde{q}_{\boldsymbol{\rm H}}\sigma_{\rm h}^{2}}.
\end{align}
The MSE performance \({\rm mse}_{\boldsymbol{\rm X}_{\rm d}}\) can be calculated as
\begin{equation}
\begin{aligned}
& {\rm mse}_{\boldsymbol{\rm X}_{\rm d}} \\
& = \int_{{x_{\rm{d}}}} {{P_{{x_{\rm{d}}}}}({x_{\rm{d}}})}  \cdot \int_{{y_{{x_{\rm{d}}}}}} {{{\left\| {{{\hat x}_{\rm{d}}} - {x_{\rm{d}}}} \right\|}^2}{P_{{y_{{x_{\rm{d}}}}}\left| {{x_{\rm{d}}}} \right.}}({y_{{x_{\rm{d}}}}}\left| {{x_{\rm{d}}}} \right.)} {\rm d}{y_{{x_{\rm{d}}}}}{\rm d}{x_{\rm{d}}}\\
&=  \frac{\sigma_{\rm x}^2}{1+\tilde{q}_{\boldsymbol{\rm X}_{\rm d}}\sigma_{\rm x}^2}.    
\end{aligned}
\end{equation}
Considering the terms \(C_{1}\) and \(C_{2}\) are difficult to be dealt with in the integral calculation, we only focus on a communication system with a large number of antennas and symbols (i.e. \(\alpha,\beta_{\rm d},\beta_{\rm p} \to \infty\)). Note that \(\tilde{q}_{\boldsymbol{\rm H}},\tilde{q}_{\boldsymbol{\rm X}_{\rm d}},\tilde{q}_{\boldsymbol{\rm X}_{\rm p}} \to \infty\). \(C_{1}\) and \(C_{2}\) can be rewritten as
\begin{equation}\label{C_1_approx}
\begin{aligned}
C_1 & = \frac{1}{\sqrt{2\pi(\tilde{q}_{\boldsymbol{\rm H}})^{-1}}}e^{-\frac{\left\|\sqrt{\tilde{q}_{\boldsymbol{\rm H}}} h + w_{h}\right\|^2}{2}} \\
& \approx \frac{1}{\sqrt{2\pi(\tilde{q}_{\boldsymbol{\rm H}})^{-1}}}e^{-\frac{\tilde{q}_{\boldsymbol{\rm H}}\left\| h \right\|^2}{2}}\\
& \approx 0, \quad\quad ({\rm as}\ \tilde{q}_{\boldsymbol{\rm H}}\to \infty)
\end{aligned}
\end{equation}
\begin{equation}\label{C_2_approx}
\begin{aligned}
C_2 & = \frac{1}{\sqrt{2\pi[\sigma_{\rm h}^2+(\tilde{q}_{\boldsymbol{\rm H}})^{-1}]}}e^{-\frac{\left\|\sqrt{\tilde{q}_{\boldsymbol{\rm H}}} h + w_{h}\right\|^2}{2(\tilde{q}_{\boldsymbol{\rm H}}\sigma_{\rm h}^2+1)}}\\
& \approx \frac{1}{\sqrt{2\pi[\sigma_{\rm h}^2+(\tilde{q}_{\boldsymbol{\rm H}})^{-1}]}}e^{-\frac{\tilde{q}_{\boldsymbol{\rm H}}\left\| h\right\|^2}{2(\tilde{q}_{\boldsymbol{\rm H}}\sigma_{\rm h}^2+1)}}\\
& \approx  \frac{1}{\sqrt{2\pi\sigma_{\rm h}^2}}e^{-\frac{\left\| h\right\|^2}{2(\sigma_{\rm h}^2+1)}}.  \quad\quad ({\rm as}\ \tilde{q}_{\boldsymbol{\rm H}}\to \infty)
\end{aligned}
\end{equation}
From \eqref{C_1_approx} and \eqref{C_2_approx}, \({\rm mse}_{\boldsymbol{\rm H}}\) can be given by
\begin{equation}
\begin{aligned}
& {\rm mse}_{\boldsymbol{\rm H}}\\
& = \int_h {{P_h}(h)}  \cdot \int_{{y_h}} {{{\left\| {\hat h - h} \right\|}^2}{P_{{y_h}\left| h \right.}}({y_h}\left| h \right.)} {\rm d}{y_h}{\rm d}h\\
& \approx \left(  \frac{\sigma_{\rm h}^2}{1+\tilde{q}_{\boldsymbol{\rm H}}\sigma_{\rm h}^2}  \right) \left( \frac{\lambda +  \tilde{q}_{\boldsymbol{\rm H}}\sigma_{\rm h}^2 }{1+\tilde{q}_{\boldsymbol{\rm H}}\sigma_{\rm h}^2}  \right).
\end{aligned}  
\end{equation}
Based on the relationship \(\tilde{q}_{\boldsymbol{\rm H}},\tilde{q}_{\boldsymbol{\rm X}_{\rm d}},\tilde{q}_{\boldsymbol{\rm X}_{\rm p}} \gg \lambda, \sigma_{\rm h}^2\), the MSE performance can be simplified as
\begin{equation}
\begin{aligned}
& {\rm mse}_{\boldsymbol{\rm X}_{\rm d}}\\
& = \frac{1}{\tilde{q}_{\boldsymbol{{\rm X}_{\rm d}}}}\\
& = \frac{1}{\alpha q_{\boldsymbol{\rm H}} {\chi}_{\rm d} }\\
& \approx \frac{\sigma_{\rm n}^2 / N}{\alpha (\lambda \sigma^2_{\rm h} - {\rm mse}_{\boldsymbol{\rm H}})} \\
& \approx \frac{\sigma_{\rm n}^2}{{\rm AE}_{\boldsymbol{\rm H}} - {\rm AMSE}_{\boldsymbol{\rm H}}},\quad\quad\quad\quad\quad\quad\quad\quad\quad\quad
\end{aligned}
\end{equation}
where \({\rm AE}_{\boldsymbol{\rm H}} = M\lambda\sigma_{\rm h}^2 = \frac{{\rm E}\{\left\|  \boldsymbol{\rm H}  \right\|^2  \}}{N}\) is the average energy of the channel matrix \(\boldsymbol{\rm H}\) per user, and \({\rm AMSE}_{\boldsymbol{\rm H}} = M \cdot {\rm mse}_{\boldsymbol{\rm H}} = \frac{{\rm E}\{\left\|  \boldsymbol{\rm H} - \hat{\boldsymbol{\rm H}}  \right\|^2  \}}{N}  \) is the average MSE of the channel matrix \(\boldsymbol{\rm H}\) per user. Then,
\begin{equation}
\begin{aligned}
& {\rm mse}_{\boldsymbol{\rm H}}\\
& = \frac{1}{\tilde{q}_{\boldsymbol{{\rm H}}}}\\
& = \frac{1}{\beta_{\rm d} q_{\boldsymbol{\rm X}_{\rm d}} {\chi }_{\rm d}+ \beta_{\rm p} q_{\boldsymbol{\rm X}_{\rm p}} {\chi }_{\rm p}} \\
& \approx \frac{\sigma_{\rm n}^2/N}{\beta_{\rm d} (\sigma_{\rm x}^2 - {\rm mse}_{\boldsymbol{\rm X}_{\rm d}})+ \beta_{\rm p} \sigma_{\rm x}^2 }\\
& \approx \frac{\sigma_{\rm n}^2}{({\rm AE}_{\boldsymbol{\rm X}_{\rm d}} - {\rm AMSE}_{\boldsymbol{\rm X}_{\rm d}}) + {\rm AE}_{\boldsymbol{\rm X}_{\rm p}}  },\quad\quad\quad\quad\quad\quad
\end{aligned}
\end{equation}
where \({\rm AE}_{\boldsymbol{\rm X}_{\rm d}} = K_{\rm d}\sigma_{\rm x}^2 = \frac{{\rm E}\{\left\|  \boldsymbol{\rm X}_{\rm d} \right\|^2  \}}{N}\) and \({\rm AE}_{\boldsymbol{\rm X}_{\rm p}} = K_{\rm p}\sigma_{\rm x}^2 = \frac{{\rm E}\{\left\|  \boldsymbol{\rm X}_{\rm p} \right\|^2  \}}{N}\) are the average energy of the symbol matrix \(\boldsymbol{\rm X}_{\rm d}, \boldsymbol{\rm X}_{\rm p}\) per user, and \({\rm AMSE}_{\boldsymbol{\rm X}_{\rm d}} = K_{\rm d} \cdot {\rm mse}_{\boldsymbol{\rm X}_{\rm d}} = \frac{{\rm E}\{\left\|  \boldsymbol{\rm X}_{\rm d} - \hat{\boldsymbol{\rm X}}_{\rm d} \right\|^2  \}}{N}\) is the average MSE of the symbol matrix \( \boldsymbol{\rm X}_{\rm d}\) per user.


%

\ifCLASSOPTIONcaptionsoff
  \newpage
\fi

\section*{}
\bibliography{LABOLABO.bib}

\end{document}